\documentclass{article}
\usepackage[a4paper]{geometry}
\usepackage{amssymb,amsmath,amsthm}
\usepackage{graphicx}
\usepackage{hyperref}
\usepackage{lineno}
\usepackage{soul,todonotes}
\usepackage{paralist}
\usepackage[capitalize]{cleveref}
\usepackage{subcaption}
\theoremstyle{plain}
\newtheorem{theorem}{Theorem}

\newtheorem{lemma}[theorem]{Lemma}

\newtheorem{observation}[theorem]{Observation}

\newtheorem{remark}{Remark}

\title{On the Deque and Rique Numbers of\\Complete and Complete Bipartite Graphs}

\author{
    Michael A. Bekos$^1$, Michael Kaufmann$^2$, Maria Eleni Pavlidi$^3$, Xenia Rieger$^4$, 
\\
\medskip
\\
\small$^1$Department of Mathematics, University of Ioannina, Ioannina, Greece\\
\small\texttt{bekos@uoi.gr}
\\
\small$^2$Institute for Computer Science, University of T{\"u}bingen, T{\"u}bingen, Germany\\
\small\texttt{michael.kaufmann@uni-tuebingen.de}
\\
\small$^3$Department of Mathematics, University of Ioannina, Ioannina, Greece\\
\small\texttt{marialenaregie3@gmail.com}
\\
\small$^4$Institute for Computer Science, University of T{\"u}bingen, T{\"u}bingen, Germany\\
\small\texttt{xenia.rieger@student.uni-tuebingen.de}
}

\date{}

\begin{document}
\thispagestyle{empty}
\maketitle
%\linenumbers

%-------------------------------------------------------------------------
\begin{abstract}
Several types of linear layouts of graphs are obtained by leveraging  known data structures; the most notable representatives are the stack and the queue layouts. In this content, given a data structure, one seeks to specify an order of the vertices of the graph and a partition of its edges into \emph{pages}, such that the endpoints of the edges assigned to each page can be processed by the given data structure in the underlying order. 

In this paper, we study deque and rique layouts of graphs obtained by leveraging the double-ended queue and the restricted-input double-ended queue (or deque and rique, for short), respectively. Hence, they generalize both the stack and the queue layouts. We focus on complete and complete bipartite graphs and present bounds on their deque- and rique-numbers, that is, on the minimum number of pages needed by any of these two types of linear layouts.  
\end{abstract}
%-------------------------------------------------------------------------

%-------------------------------------------------------------------------
\section{Introduction}
\label{sec:introduction}
%-------------------------------------------------------------------------

Stack and queue layouts form two of the most studied types of linear layouts of graphs; they date back to 70's~\cite{DBLP:journals/jct/BernhartK79,DBLP:journals/siamcomp/HeathR92} and over the years several remarkable results have been proposed in the literature~\cite{DBLP:journals/jacm/DujmovicJMMUW20,DBLP:journals/siamdm/HeathLR92,DBLP:conf/esa/0001K19,DBLP:conf/soda/JungeblutMU22,DBLP:journals/jcss/Yannakakis89}. For an introduction, refer to \cref{sec:preliminaries}. Both layouts are defined by an underlying vertex order and an edge-partition into a certain number of so-called \emph{pages} (\emph{stacks} or \emph{queues}, respectively), such that when restricting to a single page the endpoints of the edges assigned to it can be processed by the corresponding data structure in the order that appears in the underlying vertex order. 

Since given a graph the natural goal is to find a layout of it that minimizes the number of used pages under the restrictions mentioned above, stack and queue layouts have naturally been leveraged to estimate the power of the respective data structures as a mean for representing graphs (for a wealth of other applications, e.g., to VLSI design and Graph Drawing, refer to~\cite{DBLP:journals/dmtcs/DujmovicW04}). The well-known \emph{stack-number} (a.k.a.\ \emph{book-thickness} or \emph{page-number} in the literature) of a graph corresponds to the minimum number of stacks required by any of the stack layouts of it; the queue-number of a graph is defined symmetrically. In this context, it was recently shown that the stack-number of a graph cannot always be bounded by its corresponding queue-number~\cite{DBLP:journals/combinatorica/DujmovicEHMW22}, resolving a long-standing open question by Heath, Leighton and Rosenberg~\cite{DBLP:journals/siamdm/HeathLR92}; the other direction is still unknown. 

A data structure that generalizes both the stack and the queue is the so-called \emph{double-ended queue} or \emph{deque}, for short.\footnote{While in a stack insertions and removals only occur at its head and in a queue insertions only occur at its head and removals only at its tail, a deque supports insertions and removals both at its head and its tail.} As a matter of fact, the most common implementations of stacks and queues are derived by restricting corresponding implementations of deques. So, in this aspect, one naturally expects that the corresponding linear layouts that are obtained by employing the deque data structure for stipulating their edge-partitions will require fewer pages (called \emph{deques} in this content) than those of stack or queue layouts, since, obviously, the latter form a special case of the former. 

However, in contrast to the literature for stack and queue layouts, the corresponding literature for deque layouts is significantly reduced. To the best of our knowledge, there exists only one work introducing and studying deque layouts by Auer et al.~\cite{DBLP:conf/gd/AuerBBBG10}, who provide a complete characterization of the graphs admitting $1$-deque layouts (that is, deque layouts with a single deque): a graph
admits a $1$-deque layout if and only if it is a spanning subgraph of a planar graph with a Hamiltonian path; see also~\cite{DBLP:books/daglib/0035667}. Even though the \emph{deque-number} of a graph (that is, the minimum number of deques required by any of the deque layouts of the graph) has not been explicitly studied so far in the literature as a graph parameter, from the characterization by Auer et al.\ one can easily deduce the following.

\begin{observation}[Auer et al.~\cite{DBLP:conf/gd/AuerBBBG10}]\label{obs:deque-stack-bound}
The deque-number of a graph is at most half of its stack-number.
\end{observation}

Note that the queue-number is also a trivial upper bound on the deque-number of a graph. \cref{obs:deque-stack-bound}, however, immediately implies improved upper bounds on the deque-number of several graph classes, e.g., the deque-number of the complete graph $K_n$ is at most $\lceil\frac{n}{4}\rceil$~\cite{DBLP:journals/jct/BernhartK79}, the deque-number of the complete graph $K_{n,n}$ is at most $\lceil\frac{\lfloor 2n/3 \rfloor+1}{2}\rceil$~\cite{DBLP:journals/jct/EnomotoNO97}, while the deque-number of treewidth-$k$ graphs is at most $\lceil\frac{k+1}{2}\rceil$~\cite{DBLP:journals/dam/GanleyH01}. Also, since there exist maximal planar graphs that do not have a Hamiltonian path (e.g., the $n$-vertex ones with an independent set of size greater than $\frac{n}{2}+2$), it follows by a well-known result by Yannakakis~\cite{DBLP:journals/jcss/Yannakakis89} that the deque-number of planar graphs is~$2$; see also~\cite{DBLP:journals/jocg/KaufmannBKPRU20,DBLP:journals/jctb/Yannakakis20}.

Another consequence of \cref{obs:deque-stack-bound} is that deque layouts cannot be characterized by means of forbidden patterns in the underlying linear order, as it is the case, e.g., for stack and queue layouts~\cite{DBLP:journals/siamcomp/HeathR92,Oll73}; the former do not allow two edges of the same page to cross (i.e., to have alternating endpoints), while in the latter no two edges of the same page nest (i.e., have nested endpoints). The reason for the lack of such a characterization for deque layouts is the fact that maximal planar graphs with a Hamiltonian path are the maximal graphs that admit $2$-stack layouts and these layouts do not admit characterizations in terms of forbidden patterns in the underlying linear order~\cite{Wig82}. A characterization in terms of forbidden patterns is possible, however, for a special type of deque layouts, which were recently introduced and are referred to as \emph{rique layouts}~\cite{DBLP:conf/gd/BekosFKKKR22}, since the underlying data structure is of restricted input (so-called \emph{restricted-input double-ended queue} or \emph{rique},\footnote{Formally, in a rique insertions occur only at the head, and removals occur both at the head and the tail. Thus, it is a special case of a deque and a generalization of a stack or of a queue.} for short): a graph admits a $1$-rique layout if and 
only if it admits a vertex order $\prec$ avoiding three edges $(a,a')$, $(b,b')$ and $(c,c')$ such that $a \prec b \prec c \prec b' \prec \{a',c'\}$.

\paragraph{Our contribution.} In this work, we present bounds on the deque- and rique-numbers of complete and complete bipartite graphs. Especially, for deque layouts, the main research question that triggered our  work is whether it is possible to obtain better bounds than the obvious ones that one can deduce from \cref{obs:deque-stack-bound} (or in other words, whether the deque data structure is more powerful for representing graphs than two stacks). Surprisingly enough, we prove that for the case of complete graphs, this is not the case (see \cref{thm:deque-complete}), while for the case of complete bipartite graphs our upper bound shows that an improvement by a constant number is possible (to achieve this, however, we describe a rather complicated edge-to-deques assignment; see \cref{thm:deque-bipartite}). For rique layouts, our contribution is twofold. First, we improve the upper bound on the rique-number of $K_n$ from $\lceil \frac{n}{3} \rceil$~\cite{DBLP:conf/gd/BekosFKKKR22} to $\lfloor\frac{n-1}{3}\rfloor$ (see \cref{thm:rique-complete}), which we prove to be tight up to $n=30$ using an SAT-based approach (see \cref{sec:sat}). We complete our study with an upper bound of $\lfloor\frac{n-1}{2}\rfloor-1$ on the rique-number of $K_{n,n}$.

%-------------------------------------------------------------------------
\section{Preliminaries}
\label{sec:preliminaries}
%-------------------------------------------------------------------------

A \emph{vertex order} $\prec$ of a graph $G$ is a total order of its vertices, such that for any two vertices $u$ and $v$ of $G$, $u \prec v$ if and only if $u$ precedes $v$ in the order. We write $[u_1,\dots, u_k ]$ if and only if $u_i \prec u_{i+1}$ for all $1 \leq i \leq k-1$. Let $F$ be a set of $k \geq 2$ pairwise independent edges $(u_i, v_i)$ of $G$, that is, $F=\{(u_i, v_i);\;i=1,\dots,k\}$. If the order is $[u_1, \dots, u_k, v_k, \dots, v_1]$, then we say that the edges of $F$ form a \emph{$k$-rainbow}, while if the order is $[u_1, v_1, \dots, u_k, v_k]$, then the edges of $F$ form a \emph{$k$-necklace}. The edges of $F$ form a \emph{$k$-twist}, if the order is $[u_1, \dots, u_k, v_1, \dots, v_k]$; see \cref{fig:necklace-twist}. Two independent edges that form a $2$-twist ($2$-rainbow, $2$-necklace) are commonly referred to as \emph{crossing} (\emph{nested}, \emph{disjoint}, respectively). 

\begin{figure*}[t!]
	\centering	
	\begin{subfigure}[b]{.24\textwidth}
		\centering
		\includegraphics[scale=1,page=1]{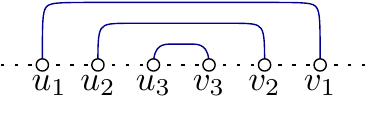}
		\caption{}
		\label{fig:rainbow}
	\end{subfigure}
	\hfil
	\begin{subfigure}[b]{.24\textwidth}
		\centering
		\includegraphics[scale=1,page=3]{figures/preliminaries}
		\caption{}
		\label{fig:necklace}
	\end{subfigure}
	\hfil
	\begin{subfigure}[b]{.24\textwidth}
		\centering
		\includegraphics[scale=1,page=2]{figures/preliminaries}
		\caption{}
		\label{fig:twist}
	\end{subfigure}
	\hfil
	\begin{subfigure}[b]{.24\textwidth}
		\centering
		\includegraphics[scale=1,page=4]{figures/preliminaries}
		\caption{}
		\label{fig:deque}
	\end{subfigure}
	\caption{Illustration of:
		(a)~a $3$-rainbow, 
		(b)~a $3$-twist,  
		(c)~a $3$-necklace, and 
        (d)~a deque.}
	\label{fig:necklace-twist}
\end{figure*}

A \emph{stack} is a set of pairwise non-crossing edges in $\prec$, while a \emph{queue} is a set of pairwise non-nested edges in $\prec$. A \emph{rique} is a set of edges in which no three edges $(a,a')$, $(b,b')$ and $(c,c')$ with $a \prec b \prec c \prec b' \prec \{a',c'\}$ exists in $\prec$. A deque is more difficult to describe due to the absence of a forbidden pattern. A relatively-simple way is the following. Assume that the vertices of a graph are arranged on a horizontal line $\ell$ from left to right according to $\prec$ (say, w.l.o.g., equidistantly). Then, each edge $(v_i,v_j)$ with $v_i \prec v_j$ can be represented either (i)~as a semi-circle that is completely above or completely below $\ell$ connecting $u_i$ and $u_j$, or (ii)~as two semi-circles on opposite sides of $\ell$, one that starts at $u_i$ and ends at a point $p_{ij}$ of $\ell$ to the right of the last vertex of $\prec$ and one that starts at point $p_{ij}$ and ends at $u_j$. With these in mind, a deque is a set of edges each of which can be represented with one of the two types (i) or (ii) that avoids crossings (such a representation is called \emph{cylindric} in~\cite{DBLP:conf/gd/AuerBBBG10}); see \cref{fig:deque}. A deque further allows classifying the edges into four categories: head-head, tail-tail, head-tail and tail-head; refer to the blue, light-blue, red and light-red edges of \cref{fig:deque}, respectively. A \emph{head-head} (\emph{tail-tail}) edge is a type-(i) edge drawn above (below, respectively) $\ell$. Symmetrically, a \emph{head-tail} (\emph{tail-head}) edge is a type-(ii) edge whose first part is above (below) $\ell$, while its second part is below (above, respectively) $\ell$. Given a deque layout $L$ and a set of edges $E$, we write $E_x$ to denote that all edges of $E$ are of type-$x$ in $L$, where $x \in \{hh,tt,ht,th\}$.

In view of the above definitions, a rique can be equivalently defined as a deque without tail-tail and tail-head edges~\cite{DBLP:conf/gd/BekosFKKKR22}. Also, it is not difficult to see that the subset of the head-head or tail-tail edges of a deque induce a stack in $\prec$, while the set of the head-tail or tail-head edges of a deque induces a queue in $\prec$. 

Since we focus on complete and complete bipartite graphs, for representing their linear layouts we use a convenient way first introduced in~\cite{DBLP:journals/jgt/MuderWW88} and subsequently used in several works~\cite{DBLP:journals/jct/EnomotoNO97,DBLP:conf/gd/FelsnerMUV21,DBLP:journals/tcs/AlamBGKP22}. Let $\prec$ be an order of the $n$ vertices $v_1,\dots,v_n$ of a graph $G$ such that $v_1 \prec \dots \prec v_n$. Then, each edge $(v_i,v_j)$ of $G$ with $i<j$ is mapped to point $(i,j)$ of the $n \times n$ grid $H=[1,n] \times [1,n]$. A set of head-head or tail-tail edges of the same page (deque or rique) corresponds to a set of points on $H$ whose union forms a monotonically decreasing curve on $H$~\cite{DBLP:journals/jgt/MuderWW88}. A set of head-tail or tail-head edges of the same page (deque or rique) corresponds to a set of points on $H$ whose union forms a monotonically increasing path on $H$~\cite{DBLP:journals/tcs/AlamBGKP22}. If a deque~contains head-tail and tail-head edges, special care is needed to avoid configurations not appearing in a cylindric layout.

%-------------------------------------------------------------------------
\section{Complete graphs}
\label{sec:complete}
%-------------------------------------------------------------------------

In this section, we study the deque- and rique-numbers of the complete graph $K_n$. As already mentioned, \cref{obs:deque-stack-bound} implies that $\lceil \frac{n}{4} \rceil$ is an easy-to-obtain upper bound on the deque-number of $K_n$. In the following, we prove that this bound is tight. To do so, we first give an estimation on the maximum number of edges that a graph admitting a $k$-deque layout can have. 

\begin{lemma}\label[lemma]{lem:deque-density}
A graph with $n$ vertices admitting a deque layout with $k$ pages has at most $(2k+1)n - 5k - 1$ edges.
\end{lemma}
\begin{proof}
Let $G$ be a graph with $n$ vertices admitting a $k$-deque layout. Let also $v_1 \prec \dots \prec v_n$ be the linear order of the vertices of $G$. Since each deque induces a planar graph, it has at most $3n-6$ edges. However, the $n-1$ so-called \emph{spine edges} $(v_i,v_{i+1})$, $ i=1,\dots,n-1$ can be added as head-head edges to every deque of the layout. So, every deque has at most $2n-5$ non-spine edges. Hence, in total $G$ has $(2n-5)k+n-1$ edges.
\end{proof}

\noindent We are now ready to prove that the deque-number of the complete graph $K_n$ is $\lceil \frac{n}{4} \rceil$.

\newcommand{\dequecomplete}{The deque-number of $K_n$ is $\lceil\frac{n}{4}\rceil$.}

\begin{theorem}\label{thm:deque-complete}
\dequecomplete
\end{theorem}
\begin{proof}
The upper bound follows from \cref{obs:deque-stack-bound} and \cite{DBLP:journals/jct/BernhartK79}. For the lower bound, let $k$ be the number of deques of $K_n$. Since $K_n$ has $\frac{n(n-1)}{2}$ edges, by \cref{lem:deque-density}, it follows that $(2k+1)n - 5k - 1 \geq \frac{n^2 - n}{2}$, which implies:

\[
k \geq \frac{n^2 - 3n + 2}{4n - 10} \quad \text{for $n$ $\geq 3$}
\]

\noindent To complete the proof of the theorem, we next show that $\lceil \frac{n^2 - 3n + 2}{4n - 10} \rceil = \lceil \frac{n}{4} \rceil$.
We do this in two steps. We first prove that $\lceil \frac{n^2 - 3n + 2}{4n - 10} \rceil - \lceil \frac{n}{4} \rceil \leq 0$.

\[
    \left\lceil \frac{n^2 - 3n + 2}{4n - 10} \right\rceil - \left\lceil \frac{n}{4} \right\rceil
    \leq \left\lceil \frac{n}{4} \cdot \frac{n - 3}{n - \frac{10}{4}} \right\rceil - \left\lceil \frac{n}{4} \right\rceil \leq 0
\]

\noindent We next prove that $\lceil \frac{n^2 - 3n + 2}{4n - 10} \rceil - \lceil \frac{n}{4} \rceil \geq 0$ holds.

\[
    \left\lceil \frac{n^2 - 3n + 2}{4n - 10} \right\rceil - \left\lceil \frac{n}{4} \right\rceil
     \geq 
    \left\lceil \frac{n}{4} \cdot \frac{n + \frac{2}{n}}{n - \frac{5}{2}} \right\rceil - \left\lceil \frac{n}{4} \right\rceil
     \geq 0
\]

\noindent Hence, the proof is completed.
\end{proof}

\noindent For rique layouts, the analog of \cref{lem:deque-density} is the following, which has been used to show a lower bound of $(1-\frac{\sqrt 2}{2})(n-2)$ on the rique-number of  $K_n$~\cite{DBLP:conf/gd/BekosFKKKR22}.

\begin{lemma}[Bekos et al.~\cite{DBLP:conf/gd/BekosFKKKR22}]\label[lemma]{lem:rique-density}
A graph with $n$ vertices admitting a rique layout with $k$ pages has at most $(2n+2)k-k^2+(n-3)$ edges.
\end{lemma}

\noindent In the next theorem, we improve the best-known upper bound on the rique-number of $K_n$ from $\lceil \frac{n}{3} \rceil$~\cite{DBLP:conf/gd/BekosFKKKR22} to $\lfloor\frac{n-1}{3}\rfloor$.

\newcommand{\riquecomplete}{The rique-number of $K_n$ is at most $\lfloor\frac{n-1}{3}\rfloor$.}
\begin{theorem}\label{thm:rique-complete}
\riquecomplete
\end{theorem}
\begin{proof}
For our proof, we distinguish three cases, namely, $n\bmod 3 \in \{0,1,2\}$. We describe each of these cases separately in the following. 

\medskip\noindent\textbf{Case 1:} $n\bmod 3 = 0$.
We start our proof assuming $n\bmod 3 = 0$ and we prove that $K_n$ admits a rique layout $\mathcal{L}$ with $\frac{n}{3}-1$ riques. Our construction contains seven ``special'' pages, namely, the ones in $\{1,2,3,4,\frac{n}{3}-3, \frac{n}{3}-2, \frac{n}{3}-1$\}; blue, red, green, dark-purple, gray, light-purple and yellow in \cref{fig:k30complete}. The remaining pages of $\mathcal{L}$ are uniform.  

\smallskip\noindent Page 1 of $\mathcal{L}$ contains the following \textcolor{blue}{$2n$} edges; see~\cref{fig:01}:
\begin{itemize}[-]
\setlength\itemsep{0em}
\item $\{(v_1,v_j), j=2,\dots, n\}_{ht}$; dark red in~\cref{fig:01}, 
\item $\{(v_i,v_n), i=2,\dots,\frac{n}{3}\}_{ht}$; red in~\cref{fig:01}, 
\item $\{(v_{\frac{n}{3}},v_j),j=\frac{n}{3}+1,\dots,\frac{2n}{3}+1\}_{hh}$; light red in~\cref{fig:01}, 
\item $\{(v_{\frac{2n}{3}+1},v_j),j =\frac{2n}{3}+2,\dots,n\}_{hh}$; blue in~\cref{fig:01}, 
\item $\{(v_{n-1},v_n)\}_{hh}$; light blue in~\cref{fig:01} .
\end{itemize}
\begin{figure}[h!]
    \centering
    \includegraphics[page=1]{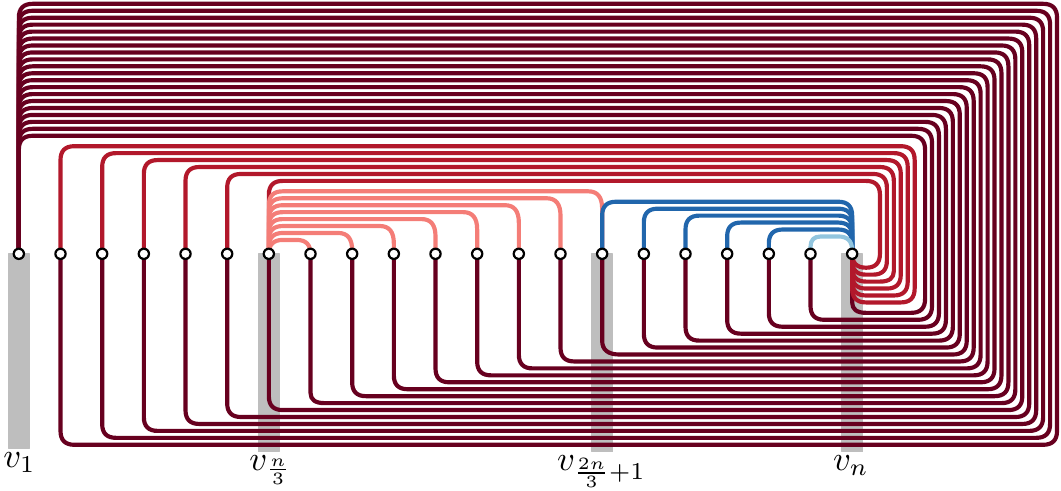}
    \caption{Page $1$ of $\mathcal{L}$ when $n\bmod 3 = 0$.}
 \label{fig:01}
\end{figure}
\smallskip\noindent Page 2 of $\mathcal{L}$ contains the following \textcolor{blue}{$2n-7$} edges: 
\begin{itemize}[-]
\setlength\itemsep{0em}
\item $\{(v_2,v_j), j=3,\dots, n-1\}_{ht}$; dark red in~\cref{fig:02},
\item $\{(v_i,v_{n-1}), i=3,\dots,\frac{n}{3}+1\}_{ht}$; red in~\cref{fig:02},
\item $\{(v_{\frac{n}{3}+1},v_n)\}_{ht}$; orange in~\cref{fig:02},
\item $\{(v_{\frac{n}{3}+1},v_j),j = \frac{n}{3}+2,\dots, \frac{2n}{3}\}_{hh}$; blue in~\cref{fig:02},
\item $\{(v_{\frac{2n}{3}},v_j),j = \frac{2n}{3}+1,\dots, n\}_{hh}$; light blue in~\cref{fig:02}.
\end{itemize}
\begin{figure*}[h!]
	\centering
	\includegraphics[page=2]{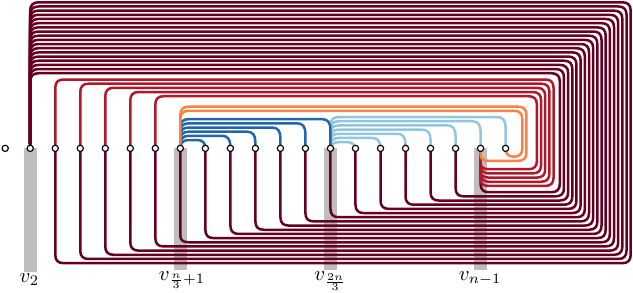}
    \caption{Page $2$ of $\mathcal{L}$ when $n\bmod 3 = 0$.}
 \label{fig:02}
\end{figure*}
\smallskip\noindent Page 3 of $\mathcal{L}$ contains the following \textcolor{blue}{$2n-5$} edges:
\begin{itemize}[-]
\setlength\itemsep{0em}
\item $\{(v_3,v_j), j=4,\dots, n-2\}_{ht}$; dark red in~\cref{fig:03},
\item $\{(v_i,v_{n-2}), i=4,\dots,\frac{n}{3}+1\}_{ht}$; red in~\cref{fig:03},
\item $\{(v_{\frac{2n}{3}+2},v_j),j = n-2,\dots, n\}_{ht}$; orange in~\cref{fig:03},
\item $\{(v_{\frac{n}{3}+1},v_{\frac{2n}{3}+1})\}_{hh}$; blue in~\cref{fig:03},
\item $\{(v_{\frac{n}{3}+2},v_j),j = \frac{2n}{3}-1,\frac{2n}{3},\frac{2n}{3}+1\}_{hh}$; light blue in~\cref{fig:03},
\item $\{(v_{\frac{n}{3}+3},v_j),j = \frac{n}{3}+4,\dots,\frac{2n}{3}-1\}_{hh}$; pink in~\cref{fig:03},
\item $\{(v_{\frac{2n}{3}+2},v_j),j = \frac{2n}{3}+3,\dots,n-3\}_{hh}$; light red in~\cref{fig:03},
\item $\{(v_{n-3},v_j),j = n-2,n-1,n\}_{hh}$; light orange in~\cref{fig:03}.
\end{itemize}
\begin{figure*}[h!]
	\centering
	\includegraphics[page=3]{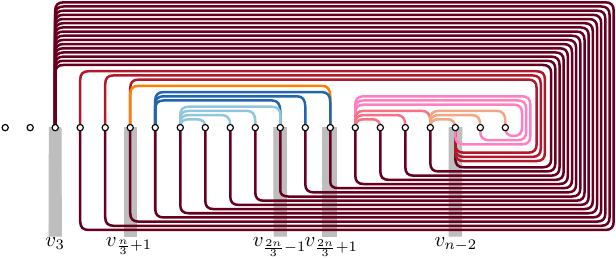}
    \caption{Page $3$ of $\mathcal{L}$ when $n\bmod 3 = 0$.}
 \label{fig:03}
\end{figure*}
\smallskip\noindent For $p=4,\dots,\frac{n}{3}-4$, page $p$ of $\mathcal{L}$ contains the following \textcolor{blue} {$\frac{n}{3}-2p+3$} edges: 
\begin{itemize}[-]
\setlength\itemsep{0em} 
\item $\{(v_p,v_j),j= p+1,\dots,n-p+1\}_{ht}$; dark red in~\cref{fig:04},
\item $\{(v_i,v_j),i=p+1,\dots \frac{n}{3}+1, j=n-p+1\}_{ht}$; red in~\cref{fig:04}, 
\item $\{(v_i,v_j),i= \frac{n}{3}+(p+1), j=n-p+1, \dots, n\}_{ht}$; pink in~\cref{fig:04}, 
\item $\{(v_i,v_j),i= \frac{n}{3}+(p+1), j=\frac{2n}{3}+(p-2), \dots, n-p\}_{hh}$; blue in~\cref{fig:04},
\item $\{(v_i,v_j),i= n-p+1, j=n-p, \dots, n\}_{hh}$; light blue in~\cref{fig:04}, 
\item $\{(v_i,v_j),i= \frac{n}{3}+(p+2), j= \frac{n}{3}+(p+3), \dots, \frac{2n}{3}+(p-2)\}_{hh}$; orange in~\cref{fig:04}.
\end{itemize}
\begin{figure*}[h!]
	\centering
	\includegraphics[page=4]{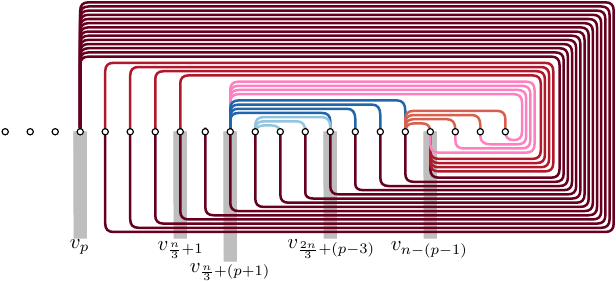}
    \caption{Page $p=4,\dots,\frac{n}{3}-4$ of $\mathcal{L}$ when $n\bmod 3 = 0$.}
 \label{fig:04}
\end{figure*}
\smallskip\noindent Page $\frac{n}{3}-3$ of $\mathcal{L}$ contains the following \textcolor{blue}{$\frac{4n}{3}+6$} edges: 
\begin{itemize}[-]
\setlength\itemsep{0em} 
\item $\{(v_{\frac{n}{3}-3},v_j), j= \frac{n}{3}-2,\dots, \frac{2n}{3}+4\}_{ht}$; dark red in~\cref{fig:05}, 
\item $\{(v_i,v_{\frac{2n}{3}+4}), i= \frac{n}{3}-2,\dots, \frac{n}{3}+1\}_{ht}$; red in~\cref{fig:05}, 
\item $\{(v_{\frac{n}{3}+3},v_j), j= \frac{2n}{3}+4,\dots, n-1\}_{ht}$; light blue in~\cref{fig:05}, 
\item $\{(v_{\frac{2n}{3}+3},v_j), j= n-1,n\}_{ht}$; pink in~\cref{fig:05},
\item $\{(v_{\frac{n}{3}+3},v_j), j= \frac{2n}{3},\dots, \frac{2n}{3}+3\}_{hh}$; dark blue in~\cref{fig:05},
\item $\{(v_{\frac{n}{3}+4},v_j), j= \frac{n}{3}+5,\dots, \frac{2n}{3}\}_{hh}$; orange in~\cref{fig:05}, 
\item $\{(v_{\frac{2n}{3}+3},v_j), j= \frac{2n}{3}+3, \dots, n-2\}_{hh}$; red in~\cref{fig:05}, 
\item $\{(v_{n-2},v_j), j= n-1,n\}_{hh}$; dark orange in~\cref{fig:05}.
\end{itemize}
\begin{figure*}[h!]
	\centering
	\includegraphics[page=5]{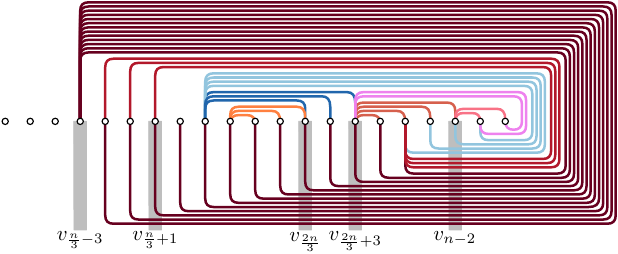}
    \caption{Page $\frac{n}{3}-3$ of $\mathcal{L}$ when $n\bmod 3 = 0$.}
 \label{fig:05}
\end{figure*}
\smallskip\noindent Page $\frac{n}{3}-2$ of $\mathcal{L}$ contains the following \textcolor{blue}{$\frac{4n}{3}+3$} edges: 
\begin{itemize}[-]
\setlength\itemsep{0em} 
\item $\{(v_{\frac{n}{3}-2},v_j), j= \frac{n}{3}-1,\dots, \frac{2n}{3}+3\}_{ht}$; dark red in~\cref{fig:06}, 
\item $\{(v_i,v_{\frac{2n}{3}+3}), i= \frac{n}{3}-1,\dots, \frac{n}{3}+1\}_{ht}$; red in~\cref{fig:06}, 
\item $\{(v_{\frac{n}{3}+2},v_j), j=\frac{2n}{3}+3,\dots, n\}_{ht}$; light  red in~\cref{fig:06}, 
\item $\{(v_{\frac{n}{3}+3},v_n)_{ht}$; pink in~\cref{fig:06}, 
\item $\{(v_{\frac{n}{3}+4},v_j), j= \frac{2n}{3}+1,\dots, n\}_{hh}$; dark orange in~\cref{fig:06}, 
\item $\{(v_{\frac{n}{3}+5},v_j), j= \frac{n}{3}+6,\dots, \frac{2n}{3}+1\}_{hh}$; orange in~\cref{fig:06}.
\end{itemize}
\begin{figure*}[h!]
	\centering
	\includegraphics[page=6]{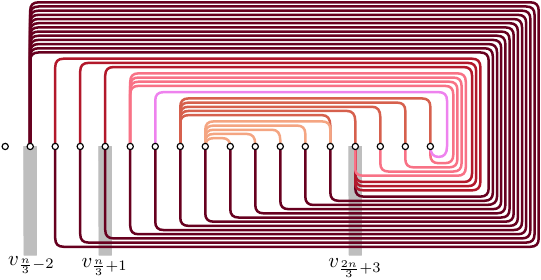}
    \caption{Page $\frac{n}{3}-2$ of $\mathcal{L}$ when $n\bmod 3 = 0$.}
 \label{fig:06}
\end{figure*}

\smallskip\noindent Page $\frac{n}{3}-1$ of $\mathcal{L}$ contains the following \textcolor{blue}{$n+9$} edges: 
\begin{itemize}[-]
\setlength\itemsep{0em} 
\item $\{(v_{\frac{n}{3}-1},v_j), j= \frac{n}{3},\dots, \frac{2n}{3}+2\}_{ht}$; dark red in~\cref{fig:07}, 
\item $\{(v_i,v_{\frac{2n}{3}+2}), i= \frac{n}{3},\dots, \frac{n}{3}+2\}_{ht}$;red in~\cref{fig:07}, 
\item $\{(v_{\frac{2n}{3}-2},v_j), j=n-5,\dots, n\}_{ht}$; light red in~\cref{fig:07}, 
\item $\{(v_{\frac{n}{3}+2},v_j), j= \frac{n}{3}+3,\dots, \frac{2n}{3}-2\}_{hh}$; orange in~\cref{fig:07}, 
\item $\{(v_{\frac{2n}{3}-1},v_j), j= \frac{2n}{3},\dots, n\}_{hh}$; light orange in~\cref{fig:07}.
\end{itemize}
 
\begin{figure*}[h!]
	\centering
	\includegraphics[page=7]{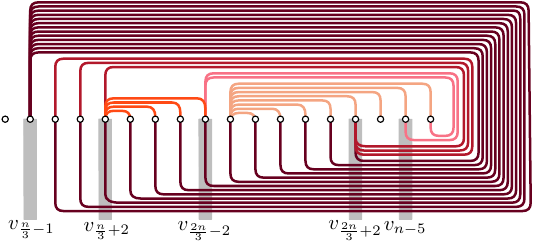}
    \caption{Page $\frac{n}{3}-1$ of $\mathcal{L}$ when $n\bmod 3 = 0$.}
 \label{fig:07}
\end{figure*}
 
\begin{figure*}[p]
    \begin{subfigure}[b]{.48\linewidth}
    \centering
    \includegraphics[scale=0.9,page=3]{figures/CompleteRique.pdf}
    \caption{}
    \label{fig:k30complete}
    \end{subfigure}
    \begin{subfigure}[b]{.48\linewidth}
    \centering
    \includegraphics[scale=0.9,page=1]{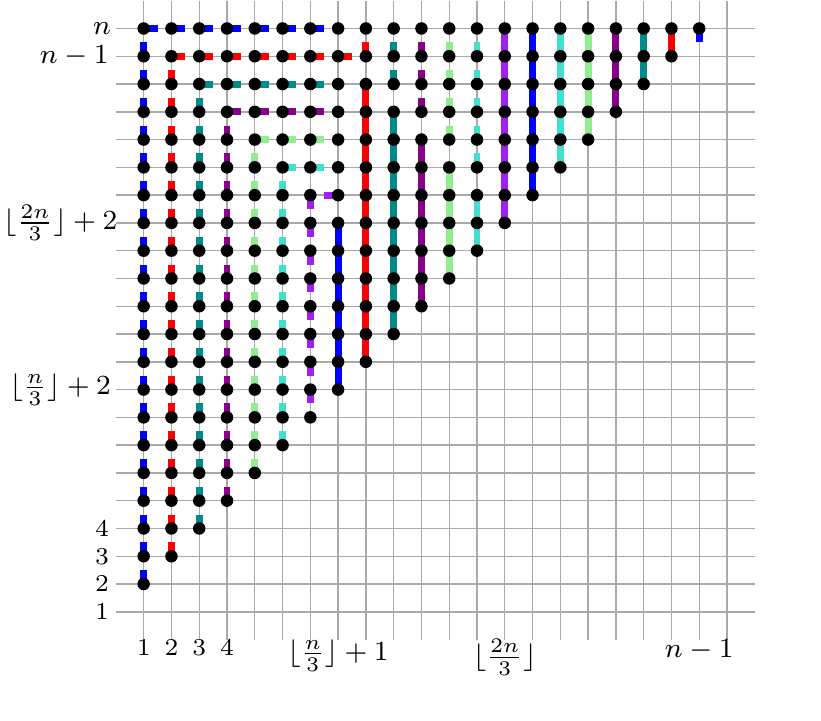}
    \caption{}
    \label{fig:k22complete}
    \end{subfigure}
    \begin{subfigure}[b]{\linewidth}
    \centering
    \includegraphics[scale=0.9,page=2]{figures/CompleteRique.pdf}
    \caption{}
    \label{fig:k23complete}
    \end{subfigure}
    \caption{Illustration of the grid representation of a rique layout of $K_{n}$ with
    (a)~$n\bmod3=0$
    (b)~$n\bmod3=1$, and 
    (c)~$n\bmod3=2$, in which paths of the same color correspond to the same rique. The points of the grid that are covered by a solid (dashed) path are head-head (head-tail, respectively).}
\end{figure*}

\noindent So, in total $\mathcal{L}$ has $(2n-1)+(2n-6)+(2n-4)+(\frac{5n}{3}-5) + \sum_{p=5}^{\frac{n}{3}-4}(\frac{5n}{3}-p+1) + (n+14)+(\frac{4n}{3}+3)+(\frac{4n}{3}+3) = \frac{n(n-1)}{2}$ edges. Since no two edges have been assigned to the same rique and all edges in the same rique form a cylindric layout, it follows that the rique number of $K_n$ is at most $\lfloor\frac{n-1}{3}\rfloor$ when $n\bmod 3 =0$.

\medskip\noindent\textbf{Case 2:} $n\bmod 3 = 1$.
In this case, we show that the $K_n$ admits a rique layout with $\lfloor \frac{n}{3} \rfloor$ riques. As in Case~1, our construction contains again ``special'' pages, namely, the ones in $\{1,2, \lfloor \frac{n-1}{3} \rfloor$\}; blue, red and purple in \cref{fig:k22complete}. The remaining pages of $\mathcal{L}$ are uniform.  

\smallskip\noindent Page $1$ of $\mathcal{L}$ contains the following \textcolor{blue}{$2n-1$} edges:
\begin{itemize}[-]
\setlength\itemsep{0em}
\item $ \{(u_1,v_j), j=2,\dots, n\}_{ht}$; dark red in~\cref{fig:11},
\item $\{(u_i,v_n), i=2,\dots, \frac{n-1}{3}+1\}_{ht}$; red in~\cref{fig:11},
\item $\{(u_{\frac{n-1}{3}+1},v_j),j = \frac{n-1}{3}+2, \dots,\frac{2n-2}{3}+1\}_{hh}$; light red in~\cref{fig:11},
\item $\{(u_{\frac{2n-2}{3}+1},v_j),j = \frac{2n-2}{3}+2, \dots,n\}_{hh}$; blue in~\cref{fig:11},
\item $\{(u_{n-1},v_{n})\}_{hh}$; light blue in~\cref{fig:11}.
\end{itemize}
 
\begin{figure*}[h!]
	\centering
	\includegraphics[page=1]{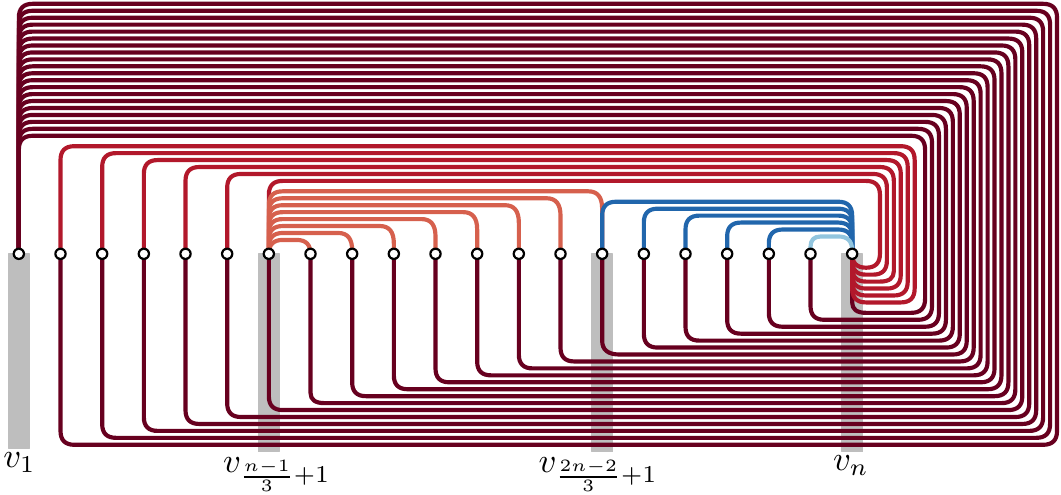}
    \caption{Page $1$ of $\mathcal{L}$ when $n\bmod 3 = 1$.}
 \label{fig:11}
\end{figure*}

\begin{figure*}[b!]
	\centering
	\includegraphics[page=2]{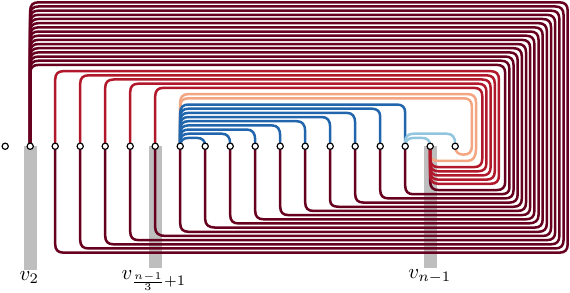}
    \caption{Page $2$ of $\mathcal{L}$ when $n\bmod 3 = 1$.}
 \label{fig:12}
\end{figure*}

\smallskip\noindent Page $2$ of $\mathcal{L}$ contains the following \textcolor{blue}{$2n-4$} edges:
\begin{itemize}[-]
\setlength\itemsep{0em}
\item $\{(u_2,v_j), j=3,\dots, n-1\}_{ht}$; dark red in~\cref{fig:12},
\item $\{(u_i,v_{n-1}), i=3,\dots,\frac{n-1}{3}+1\}_{ht}$; red in~\cref{fig:12},
\item $\{(u_{\frac{n-1}{3}+2},v_j),j=n-1,n\}_{ht}$; light red in~\cref{fig:12},
\item $\{(u_{\frac{n-1}{3}+2},v_j),j = \frac{n-1}{3}+3 \dots, n-2\}_{hh}$; blue in~\cref{fig:12},
\item $\{(u_{n-2},v_j),j=n-1,n\}_{hh}$; light blue in~\cref{fig:12}.
\end{itemize}

\smallskip\noindent For $p=3,\dots,\frac{n-1}{3}-1$, page $p$ of $\mathcal{L}$ contains the following \textcolor{blue} {$2n-3p+2$} edges: 
\begin{itemize}[-]
\setlength\itemsep{0em} 
\item $\{(u_p,v_j), j= p+1,\dots, n-p+1\}_{ht}$; dark red in~\cref{fig:1p}, 
\item $\{(u_i,v_{n-p+1}),i=p+1,\dots,\frac{n-1}{3} +1\}_{ht}$; red in~\cref{fig:1p},
\item $\{(u_{\frac{n-1}{3} +p},v_j), j=n-p+1, \dots, n\}_{ht}$; orange in~\cref{fig:1p},
\item $\{(u_{\frac{n-1}{3} +p},v_j), j= \frac{n-1}{3}  + (p+1), \dots, n-p\}_{hh}$; blue in~\cref{fig:1p},
\item $\{(u_{n-p},v_j),j=n-p+1, \dots, n\}_{hh}$; light blue in~\cref{fig:1p}.
\end{itemize}
 
\begin{figure*}[h!]
	\centering
	\includegraphics[page=3]{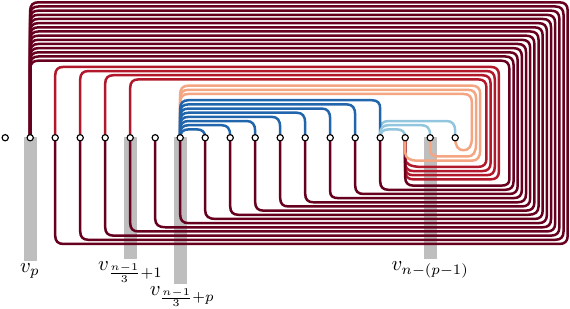}
    \caption{Page $p=3,\dots,\frac{n-1}{3}-1$ of $\mathcal{L}$ when $n\bmod 3 = 1$.}
 \label{fig:1p}
\end{figure*}
 
\smallskip\noindent Page $\lfloor \frac{n-1}{3} \rfloor$ of $\mathcal{L}$ contains the following \textcolor{blue}{$2(\frac{n-1}{3})+4$} edges:
\begin{itemize}[-]
\setlength\itemsep{0em} 
\item $\{(u_{ \frac{n-1}{3} },v_j), j=\frac{n-1}{3}+1,\dots, \frac{2n-2}{3}+2)\}_{ht}$; dark red in~\cref{fig:14}, 
\item $\{(u_{\frac{n-1}{3}+1},v_{\frac{2n-2}{3}+2}),\}_{ht}$; red in~\cref{fig:14},
\item $\{(u_{ \frac{2n-2}{3} },v_j),j= \frac{2n-2}{3}+1 , \dots, n\}_{hh}$; orange in~\cref{fig:14}.
\end{itemize}
 
\begin{figure*}[h!]
	\centering
	\includegraphics[page=4]{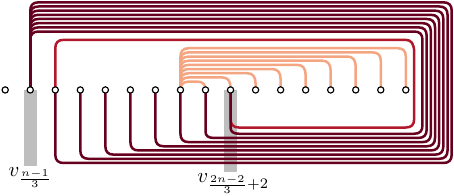}
    \caption{Page $\lfloor\frac{n-1}{3}\rfloor$ of $\mathcal{L}$ when $n\bmod 3 = 1$.}
 \label{fig:14}
\end{figure*}
 
\noindent So, when $n\bmod3 = 1$,  $\mathcal{L}$ has $2n-1+2n-4+\sum_{p=3}^{ \frac{n-1}{3} -1}(2n-3p+2)+2(\frac{n-1}{3})+4 = \frac{n(n-1)}{2}$ edges. Since no two edges have been assigned to the same rique and all edges in the same rique form a cylindric layout, it follows that the rique number of $K_n$ is at most $\lfloor\frac{n-1}{3}\rfloor$ when $n\bmod3 =1$.

\medskip\noindent\textbf{Case 3:} $n\bmod 3 = 2$. We continue with the case $n\bmod3 = 2$. In this case, we show that the $K_n$ admits a rique layout with $\lfloor \frac{n}{3} \rfloor$ riques. As with the previous two cases, our construction contains again ``special'' pages, namely, the ones in $\{1,2$\}; blue and red in \cref{fig:k23complete}. The remaining pages of $\mathcal{L}$ are uniform. 

\smallskip\noindent Page $1$ of $\mathcal{L}$ contains the following \textcolor{blue}{$2n$} edges:
\begin{itemize}[-]
\setlength\itemsep{0em}
\item $ \{(u_1,v_j), j=2,\dots, n\}_{ht}$; dark red in~\cref{fig:21},
\item $\{(u_i,v_n), i=2,\dots, \frac{n-2}{3}+1\}_{ht}$; red in~\cref{fig:21},
\item $\{(u_{\frac{n-2}{3}+1},v_j),j = \frac{n-2}{3}+2, \dots,\frac{2n-4}{3}+2\}_{hh}$; light red in~\cref{fig:21},
\item $\{(u_{\frac{2n-4}{3}+2},v_j),j =\frac{2n-4}{3}+3, \dots,n\}_{hh}$; blue in~\cref{fig:21},
\item $\{(u_{n-1},v_{n})\}_{hh}$; light blue in~\cref{fig:21}.
\end{itemize}
 
\begin{figure*}[h!]
	\centering
	\includegraphics[page=1]{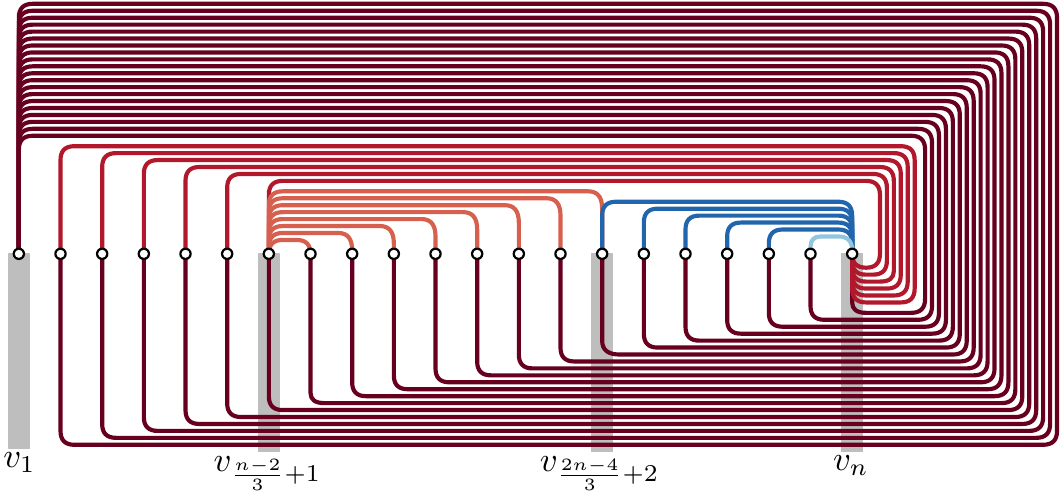}
    \caption{Page $1$ of $\mathcal{L}$ when $n\bmod 3 = 2$.}
 \label{fig:21}
\end{figure*}
 
\begin{figure*}[b!]
	\centering
	\includegraphics[page=2]{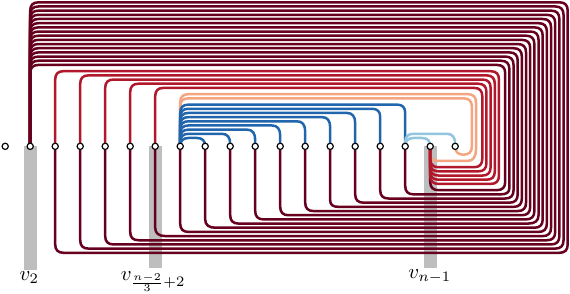}
    \caption{Page $2$ of $\mathcal{L}$ when $n\bmod 3 = 2$.}
 \label{fig:22}
\end{figure*}
 
\smallskip\noindent Page $2$ of $\mathcal{L}$ contains the following \textcolor{blue}{$2n-4$} edges:
\begin{itemize}[-]
\setlength\itemsep{0em}
\item $\{(u_2,v_j), j=3,\dots, n-1\}_{ht}$; dark red in~\cref{fig:22},
\item $\{(u_i,v_{n-1}), i=3,\dots,\frac{n-2}{3}+2\}_{ht}$; red in~\cref{fig:22},
\item $\{(u_{\frac{n-2}{3}+2},v_j),j=n-1,n\}_{ht}$; light red in~\cref{fig:22},
\item $\{(u_{\frac{n-2}{3}+2},v_j),j = \frac{n-2}{3}+3 \dots, n-2\}_{hh}$; blue in~\cref{fig:22},
\item $\{(u_{n-2},v_j),j=n-1,n\}_{hh}$; light blue in~\cref{fig:22}.
\end{itemize}

\smallskip\noindent For $p=3,\dots,\frac{n-2}{3}$, page $p$ of $\mathcal{L}$ contains the following \textcolor{blue} {$2n-3p+2$} edges: 
\begin{itemize}[-]
\setlength\itemsep{0em} 
\item $\{(u_p,v_j), j= p+1,\dots, n-p+1\}_{ht}$; dark red in~\cref{fig:2p}, 
\item $\{(u_i,v_{n-p+1}),i=p+1,\dots,  \frac{n-2}{3} +1\}_{ht}$; red in~\cref{fig:2p},
\item $\{(u_{\frac{n-2}{3} +p},v_j), j=n-p+1, \dots, n\}_{ht}$;light red in~\cref{fig:2p},
\item $\{(u_{\frac{n-2}{3} +p},v_j), j=  \frac{n-2}{3} + (p+1), \dots, n-p\}_{hh}$; blue in~\cref{fig:2p},
\item $\{(u_{n-p},v_j),j=n-p+1, \dots, n\}_{hh}$; light blue in~\cref{fig:2p}.
\end{itemize}
 
\begin{figure*}[h!]
	\centering
	\includegraphics[page=3]{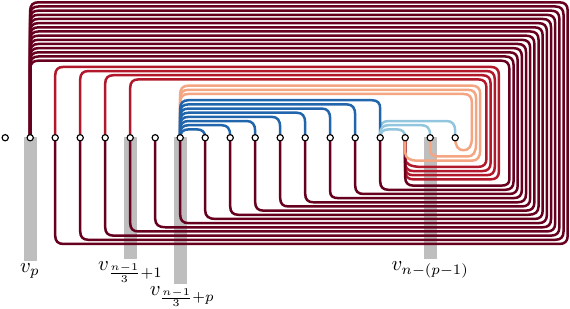}
    \caption{Page $p=3,\dots,\frac{n-2}{3}$ of $\mathcal{L}$ when $n\bmod 3 = 2$.}
 \label{fig:2p}
\end{figure*}
 
\noindent So, in total $\mathcal{L}$ has $2n+2n-4+\sum_{p=3}^{ \frac{n-2}{3} }(2n-3p+2) = \frac{n(n-1)}{2}$ edges. Since no two edges have been assigned to the same rique and all edges in the same rique form a cylindric layout, it follows that the rique number of $K_n$ is at most $\lfloor\frac{n-1}{3}\rfloor$ when $n\bmod 3 =2$.
\end{proof}

\begin{remark}\label{rem:bound-tightness}\normalfont
Using the SAT formulation that we present in \cref{sec:sat} we were able to show that the upper bound of \cref{thm:rique-complete} is tight for all values of $n \leq 30$. However, we were not able to show a matching lower bound. In view of these observations, we conjecture in \cref{sec:conclusions} that the rique-number of $K_n$ is exactly $\lfloor\frac{n-1}{3}\rfloor$.
\end{remark}

\section{Complete bipartite graphs}
\label{sec:bipartite}

In this section, we study the deque- and rique-numbers of the complete bipartite graph $K_{n,n}$. Let the two parts of $K_{n,n}$ be $A = \{a_1, \dots, a_n\}$ and $B = \{b_1,\dots,b_n\}$ with $|A| = |B| = n$. W.l.o.g., we may assume that in the computed layouts $a_1 \prec \dots \prec a_n$ and $b_1 \prec \dots \prec b_n$ holds.

\clearpage

\newcommand{\dequebipartite}{The deque-number of $K_{n,n}$ is at most $\lceil\frac{n}{3}\rceil$.}

\begin{theorem}\label{thm:deque-bipartite}
\dequebipartite
\end{theorem}

\begin{proof}
Assume that $n\bmod3 = 0$; the remaining cases follow from this one. We describe a deque layout $\mathcal{L}$ of $K_{n,n}$ with $\frac{n}{3}$ deques, in which the underlying order is:
$a_1 \prec \dots \prec a_{n/3} \prec b_1 \dots \prec b_{2n/3} \prec a_{n/3+1} \prec \dots \prec a_n \prec b_{2n/3+1} \prec \dots \prec b_n$; for an illustration of the corresponding matrix representation refer to~\cref{fig:bipartite-matrix-deque} and observe that it consists of four blocks, e.g., the middle block corresponds to the edges between $a_{n/3+1},\dots, a_{n}$ and $b_1,\dots,b_{2n/3}$.

\smallskip\noindent Page 1 of $\mathcal{L}$ contains the following \textcolor{blue}{$2n+16$} edges:
 
\begin{itemize}[-]
\setlength\itemsep{0em}
    \item $\{(a_1, b_j), j = \frac{2n}{3} + 4, \dots, n\}_{hh}$; dark red in~\cref{fig:d1}, 
    \item $\{(a_1, b_6)\}_{hh}$; red in~\cref{fig:d1},
    \item $\{(a_2, b_j), j = 1, \dots, 6\}_{hh}$; light red in~\cref{fig:d1},
    \item $\{(a_{\frac{n}{3} + 8}, b_j), j = 8, 9, 10\}_{hh}$; orange in~\cref{fig:d1},
    \item $\{(a_i, b_{10}), i = \frac{n}{3} + 1, \dots, \frac{n}{3} + 7\}_{hh}$; light orange in~\cref{fig:d1},
    \item $\{(a_i, b_{\frac{2n}{3} + 4}), i = \frac{n}{3} + 8, \dots, n\}_{hh}$; green in~\cref{fig:d1},
    \item $\{(a_2, b_j), j = n - 2, \dots, n\}_{tt}$; dark blue in~\cref{fig:d1},
    \item $\{(a_i, b_{\frac{2n}{3} + 5}), i = \frac{n}{3} - 2,\dots, \frac{n}{3}\}_{tt}$; blue in~\cref{fig:d1},
    \item $ \{(a_{\frac{n}{3}}, b_j), j = \frac{2n}{3} + 1, \frac{2n}{3} + 2)\}_{tt}$; light blue in~\cref{fig:d1}.
    \item $\{(a_n, b_1), (a_n, b_2), (a_{n - 1}, b_2)\}_{tt}$; pink in~\cref{fig:d1},
    \item $\{(a_i, b_j), (a_{i - 1}, b_j), i = n - 1, \dots, \frac{2n}{3}+3, j = 4, \dots, \frac{n}{3}\}_{tt}$; light pink in~\cref{fig:d1},
    \item $\{(a_i,b_{\frac{n}{3}}), i = \frac{2n}{3}+1,\frac{2n}{3}+2\}_{tt}$; dark orange in~\cref{fig:d1},
    \item $\{(a_{\frac{2n}{3}+1},b_j), j = \frac{n}{3}+1, \dots,\frac{n}{3}+3\}_{tt}$; dark green in~\cref{fig:d1},
    \item $\{(a_i,b_{\frac{n}{3}+3}), i = \frac{n}{3}+1, \dots,\frac{2n}{3}\}_{tt}$; green in~\cref{fig:d1}.
\end{itemize}

\begin{figure*}[h!]
	\centering	
    \includegraphics[page=5,width=\textwidth]{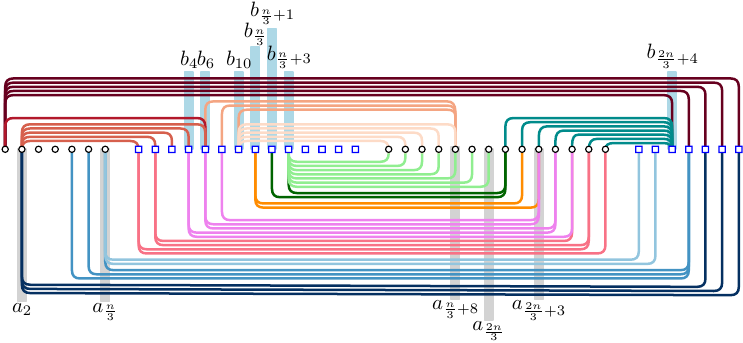}
    \caption{Page $1$ of $\mathcal{L}$.}
 \label{fig:d1}
\end{figure*}

\clearpage

\smallskip\noindent Page 2 of $\mathcal{L}$ contains the following \textcolor{blue}{$\frac{8n}{3}+5$} edges:

\begin{itemize}[-]
\setlength\itemsep{0em}
    \item $\{(a_1, b_{\frac{2n}{3} +3})\}_{ht}$; black in~\cref{fig:d2},
    \item $\{(a_2, b_j), j = \frac{2n}{3} + 3, \dots, n -3\}_{ht}$; dark red in~\cref{fig:d2},
    \item $\{(a_3, b_j), j = n - 3, \dots, n-1\}_{ht}$; red in~\cref{fig:d2},
    \item $\{(a_i, b_n), i = 3, \dots, \frac{n}{3}\}_{ht}$; dark orange in~\cref{fig:d2},
    \item $\{(a_i,b_1), i= n-3,n-4\}_{hh}$; light orange in~\cref{fig:d2},
    \item $\{(a_i,b_2), i= n-4,n-5\}_{hh}$; dark orange in~\cref{fig:d2},
    \item $\{(a_i,b_j),(a_{i-1},b_j) i= n-5,\dots,\frac{2n}{3}+1, j = 4,\dots, \frac{n}{3}-2\}_{hh}$; dark pink in~\cref{fig:d2},
    \item $\{(a_i,b_{\frac{n}{3}-2}), i=\frac{2n}{3}+1,\dots,\frac{2n}{3}-1\}_{hh}$; light pink in~\cref{fig:d2},
    \item$\{(a_{\frac{2n}{3}-1},b_j),j=\frac{n}{3}-1,\frac{n}{3}\}_{hh}$; light gray in~\cref{fig:d2},
    \item$\{(a_i,b_{\frac{n}{3}+1}), i=\frac{n}{3}+1,\dots,\frac{2n}{3}-1\}_{hh}$; gray in~\cref{fig:d2},
    \item $\{(a_i, b_n), i = n - 3, \dots, n\}_{hh}$; pink in~\cref{fig:d2}, 
    \item $\{(a_1, b_i), i = 1, \dots, 5\}_{tt}$; light green in~\cref{fig:d2}, 
    \item $\{(a_i,b_j),(a_i,b_{j+1}), i = n-1,\dots, \frac{2n}{3}+4, j = 5,\dots,\frac{n}{3}\}_{tt}$; green in~\cref{fig:d2},
    \item $\{(a_i,b_{\frac{n}{3}+1}), i = \frac{2n}{3}+2,\frac{2n}{3}+3\}_{tt}$; dark blue in~\cref{fig:d2},
    \item $\{(a_{\frac{2n}{3}+2},b_j), i = \frac{n}{3}+2,\frac{n}{3}+3\}_{tt}$; blue in~\cref{fig:d2},
    \item $\{(a_i,b_{\frac{n}{3}+4}), i = \frac{n}{3}+1,\dots,\frac{2n}{3}+2\}_{tt}$; light blue in~\cref{fig:d2}.
\end{itemize} 

\begin{figure*}[ h!]
	\centering	
    \includegraphics[page=7,width=\textwidth]{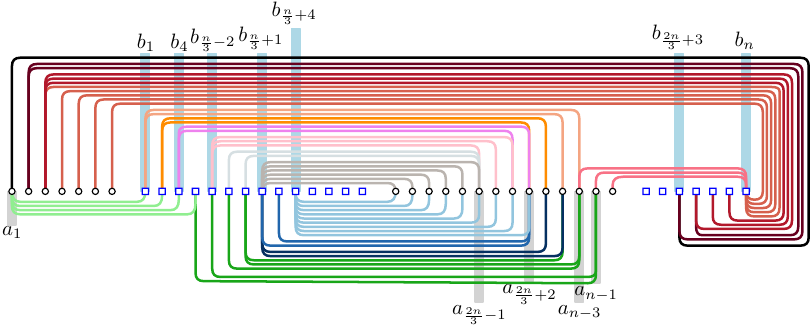}
    \caption{Page $2$ of $\mathcal{L}$.}
 \label{fig:d2}
\end{figure*}

\clearpage

\smallskip\noindent Page $p=3,4$ of $\mathcal{L}$ contains the following \textcolor{blue}{$\frac{8n}{3}+p+5$} edges:

\begin{itemize}[-]
\setlength\itemsep{0em}
    \item $\{(a_i, b_{2p+1}),i=1,\dots, p-1\}_{tt}$; green in~\cref{fig:dp},
    \item $\{(a_p, b_j), j = 1, \dots, 2p+1\}_{tt}$; dark green in~\cref{fig:dp},
    \item $\{(a_i, b_{\frac{2n}{3} + 5 - p}), i = 1, \dots, p-1 )\}_{ht}$; dark red in~\cref{fig:dp},
    \item $\{(a_p, b_j), j = \frac{2n}{3} + 5 - p, \dots, n - 1 - p\}_{ht}$; red in~\cref{fig:dp},
    \item $\{(a_{p + 1}, b_j), j = n - 1 - p, \dots, n + 1 - p\}_{ht}$; pink in~\cref{fig:dp},
    \item $\{(a_i, b_{n + 2 - p}), i = 1 + p, \dots, \frac{n}{3}\}_{ht}$; light pink in~\cref{fig:dp},
    \item $\{(a_i,b_1), i= n-2p+1,n-2p\}_{hh}$; purple in~\cref{fig:dp},
    \item $\{(a_i,b_2), i= n-2p,n-2p-1\}_{hh}$; orange in~\cref{fig:dp},
    \item $\{(a_i,b_j),(a_{i-1},b_j) i= n-2p-1,\dots,\frac{2n}{3}+4-p , j = 4,\dots, \frac{n}{3}-p-1\}_{hh}$; yellow in~\cref{fig:dp},
    \item $\{(a_i,b_{\frac{n}{3}-p}), i=\frac{2n}{3}-p+3,\dots,\frac{2n}{3}-p+1\}_{hh}$; dark pink in~\cref{fig:dp},
    \item$\{(a_{\frac{2n}{3}-p},b_j),j=\frac{n}{3}-p+1,\frac{n}{3}-p+2 \}_{hh}$; light yellow in~\cref{fig:dp},
    \item$\{(a_i,b_{\frac{n}{3}-p+3}), i=\frac{n}{3}+1,\dots,\frac{2n}{3}-p+1\}_{hh}$; dark yellow in~\cref{fig:dp},
    \item $\{(a_i,b_j),(a_i,b_{j+1}), i = n-1,\dots, \frac{2n}{3}+p+2, j = 5,\dots,\frac{n}{3}+p-2\}_{tt}$, light green in~\cref{fig:dp},
    \item $\{(a_i,b_{\frac{n}{3}+p-1}), i = \frac{2n}{3}+p,\frac{2n}{3}+p+1\}_{tt}$;blue in~\cref{fig:dp},
    \item $\{(a_{\frac{2n}{3}+p},b_j), i = \frac{n}{3}+p,\frac{n}{3}+p+1\}_{tt}$; dark blue in~\cref{fig:dp},
   \item $\{(a_i,b_{\frac{n}{3}+p+2}), i = \frac{n}{3}+1,\dots,\frac{2n}{3}+p\}_{tt}$; light blue in~\cref{fig:dp},
    \item $\{(a_{n + 2 - 2p}, b_j), j = n + 2 - p, \dots, n\}_{hh}$; light gray in~\cref{fig:dp},
    \item $\{(a_i, b_{n + 2 - p}), i = n + 3 - 2p, \dots, n\}_{hh}$; gray in~\cref{fig:dp},
    \item $\{(a_{n + 1 - 2p}, b_j), j = n + 2 - p, \dots, n\}_{ht}$; black in~\cref{fig:dp}.
\end{itemize} 

\begin{figure*}[ h!]
	\centering	
    \includegraphics[page=9,width=\textwidth]{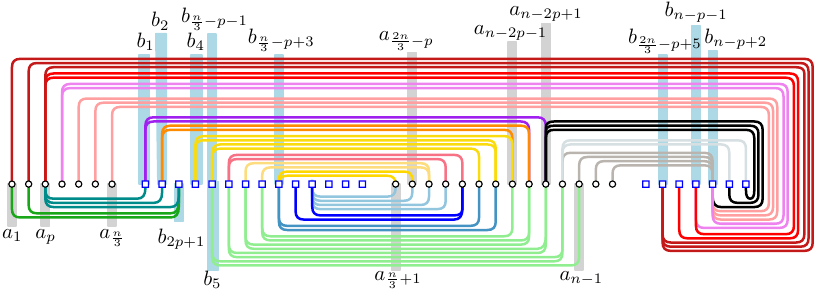}
    \caption{Page $p=3,4$ of $\mathcal{L}$.}
 \label{fig:dp}
\end{figure*}

\clearpage

\smallskip\noindent For $p =5, \dots, \frac{n}{3}-8$, page $p$ of $\mathcal{L}$ contains \textcolor{blue}{$\frac{8n}{3}+2p+18$} edges:

\begin{itemize}[-]
\setlength\itemsep{0em}
    \item $\{(a_i,b_1), i= n-2p+1,n-2p\}_{hh}$; purple in~\cref{fig:dp2},
    \item $\{(a_i,b_2), i= n-2p,n-2p-1\}_{hh}$; orange in~\cref{fig:dp2},
    \item $\{(a_i,b_j),(a_{i-1},b_j) i= n-2p-1,\dots,\frac{2n}{3}+4-p , j = 4,\dots, \frac{n}{3}-p-1\}_{hh}$; yellow in~\cref{fig:dp2},
    \item $\{(a_i,b_{\frac{n}{3}-p}), i=\frac{2n}{3}-p+3,\dots,\frac{2n}{3}-p+1\}_{hh}$; dark pink in~\cref{fig:dp2},
    \item$\{(a_{\frac{2n}{3}-p},b_j),j=\frac{n}{3}-p+1,\frac{n}{3}-p+2 \}_{hh}$; light yellow in~\cref{fig:dp2},
    \item$\{(a_i,b_{\frac{n}{3}-p+3}), i=\frac{n}{3}+1,\dots,\frac{2n}{3}-p+1\}_{hh}$; dark yellow in~\cref{fig:dp2},   
    \item $\{(a_i,b_j),(a_i,b_{j+1}), i = n-1,\dots, \frac{2n}{3}+p+2, j = 5,\dots,\frac{n}{3}+p-2\}_{tt}$; light green in~\cref{fig:dp2},
    \item$\{(a_i,b_{\frac{n}{3}+p+2}), i = \frac{n}{3}+1,\dots, \frac{2n}{3}+p \}_{tt}$; light blue in~\cref{fig:dp2},
    \item$\{(a_{\frac{2n}{3}+p},b_j), j = \frac{n}{3}+p-1,\dots,\frac{n}{3}+p+1\}_{tt}$; blue in~\cref{fig:dp2},
    \item$\{(a_{\frac{2n}{3}+p+1},b_{\frac{n}{3}+p-1})\}_{tt}$; dark blue in~\cref{fig:dp2},
    \item $\{(a_p, b_j), i = 1, \dots, 2p -1\}_{tt}$; dark green in~\cref{fig:dp2},
    \item $\{(a_{p + 1}, b_j), j = 1, \dots, p-1\}_{tt}$; green in~\cref{fig:dp2},
    \item $\{(a_i, b_{2p}), i = 1, \dots, p\}_{ht}$; dark red in~\cref{fig:dp2},       
    \item $\{(a_{p}, b_j), j = \frac{2n}{3} + 1, \dots, n - 1 - p\}_{ht}$; red in~\cref{fig:dp2},
    \item $\{(a_{p + 1}, b_j), j = n - 1 - p, \dots, n + 1 - p\}_{ht}$; pink in~\cref{fig:dp2},
    \item $\{(a_i, b_{n + 2 - p}), i = p + 1, \dots, \frac{n}{3}\}_{ht}$; light pink in~\cref{fig:dp2},
    \item $\{(a_{n + 2 - 2p}, b_j), j = n + 2 - p, \dots, n\}_{hh}$; light gray in~\cref{fig:dp2},
    \item $\{(a_i, b_{n + 2 - p}), i = n + 3 - 2p, \dots, n\}_{hh}$; gray in~\cref{fig:dp2},
    \item $\{(a_{n + 1 - 2p}, b_j), j = n + 2 - p, \dots, n\}_{ht}$; black in~\cref{fig:dp2}.
    
\end{itemize} 

\begin{figure*}[h!]
	\centering	
    \includegraphics[page=8,width=\textwidth]{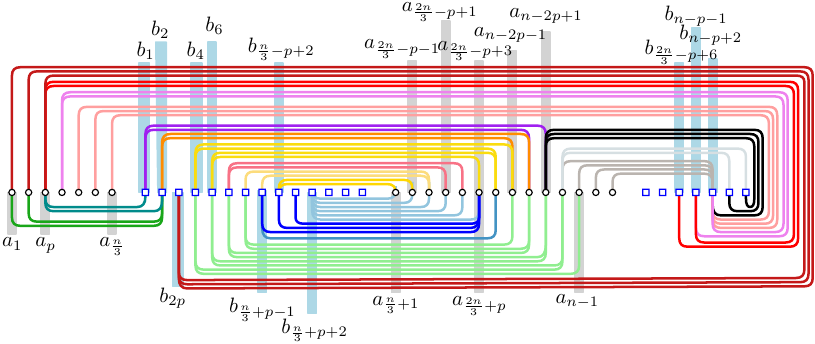}
    \caption{Page $p =5, \dots, \frac{n}{3}-8$ of $\mathcal{L}$.}
 \label{fig:dp2}
\end{figure*}

\clearpage

\smallskip\noindent For $p =\frac{n}{3}-7, \dots, \frac{n}{3}-4$, page $p$ of $\mathcal{L}$ contains \textcolor{blue}{$\frac{8n}{3}-2p-1$} edges:

\begin{itemize}[-]
\setlength\itemsep{0em}
    \item$\{(a_i,b_1), i = n-2p+1,n-2p\}_{hh}$; purple in~\cref{fig:dp4},
    \item$\{(a_i,b_2), i = n-2p,n-2p-1\}_{hh}$; dark orange in~\cref{fig:dp4},
    \item$\{(a_i,b_j),(a_{i-1},b_j), i=n-2p-1,\dots, n-2p+3, j=4,\dots,6\}_{hh}$; yellow in~\cref{fig:dp4},
    \item$\{(a_i,b_{\frac{n}{3}-p}), i = \frac{2n}{3}-p+1,\dots, \frac{2n}{3}-p+3\}_{hh}$ ; dark pink in~\cref{fig:dp4},  
    \item$\{(a_{\frac{2n}{3}-p},b_j), j = \frac{n}{3}-p,\frac{n}{3}-p+1\}_{hh}$; light orange in~\cref{fig:dp4},
    \item $\{(a_i, b_{\frac{n}{3}-p+2}), i = \frac{n}{3}+1,\dots, \frac{2n}{3}-p\}_{hh}$; orange in~\cref{fig:dp4},  
    \item$\{(a_i,b_j),(a_i,b_{j+1}),i=n-1,\dots,n-5,j= 2p ,\dots, \frac{n}{3}+p-2\}_{tt}$; light green in~\cref{fig:dp4},   
    \item$\{(a_i,b_{\frac{n}{3}+p+2}), i = \frac{n}{3}+1,\dots, \frac{2n}{3}+p \}_{tt}$; light blue in~\cref{fig:dp4},
    \item$\{(a_{\frac{2n}{3}+p},b_j), j = \frac{n}{3}+p-1,\dots,\frac{n}{3}+p+1\}_{tt}$; blue in~\cref{fig:dp4},
    \item$\{(a_{\frac{2n}{3}+p+1},b_{\frac{n}{3}+p-1})\}_{tt}$; dark blue in~\cref{fig:dp4},
    \item $\{(a_p, b_j), i = 1, \dots, 2p -1\}_{tt}$; dark green in~\cref{fig:dp4},
    \item $\{(a_{p + 1}, b_j), j = 1, \dots, p-1\}_{tt}$; green in~\cref{fig:dp4},
    \item $\{(a_i, b_{2p}), i = 1, \dots, p\}_{ht}$; dark red in~\cref{fig:dp4},
    \item $\{(a_{p}, b_j), j = \frac{2n}{3} + 1, \dots, n - 1 - p\}_{ht}$; red in~\cref{fig:dp4},
    \item $\{(a_{p + 1}, b_j), j = n - 1 - p, \dots, n + 1 - p\}_{ht}$; pink in~\cref{fig:dp4},
    \item $\{(a_i, b_{n + 2 - p}), i = p + 1, \dots, \frac{n}{3}\}_{ht}$; light pink in~\cref{fig:dp4},
    \item $\{(a_{n + 2 - 2p}, b_j), j = n + 2 - p, \dots, n\}_{hh}$; light gray in~\cref{fig:dp4},
    \item $\{(a_i, b_{n + 2 - p}), i = n + 3 - 2p, \dots, n\}_{hh}$; gray in~\cref{fig:dp4},
    \item $\{(a_{n + 1 - 2p}, b_j), j = n + 2 - p, \dots, n\}_{ht}$; black in~\cref{fig:dp4}.
\end{itemize}

\begin{figure*}[ h!]
	\centering	
    \includegraphics[page=8,width=\textwidth]{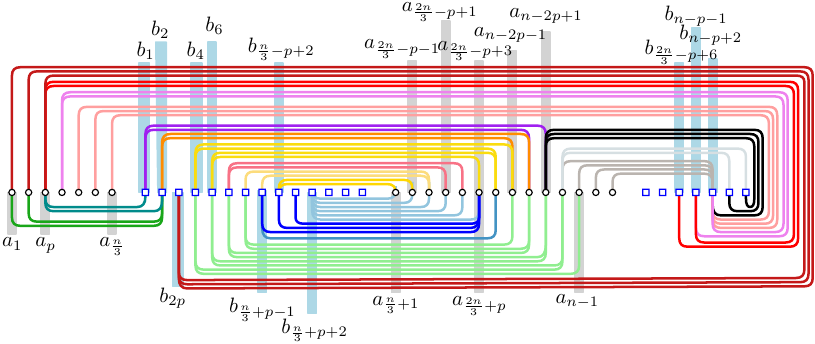}
    \caption{Page $p =\frac{n}{3}-7, \dots, \frac{n}{3}-4$ of $\mathcal{L}$.}
 \label{fig:dp4}
\end{figure*}

\clearpage

\smallskip\noindent Page $\frac{n}{3} - 3$ of $\mathcal{L}$ contains the following \textcolor{blue}{$\frac{10n}{3}+11$} edges:
 
\begin{itemize}[-]
\setlength\itemsep{0em}
    \item $\{(a_i, b_{\frac{2n}{3}-6}), i = 1, \dots, \frac{n}{3} - 3 \}_{ht}$; dark red in~\cref{fig:d4},
    \item $\{(a_{\frac{n}{3} - 3}, b_j), j=\frac{2n}{3} + 1,\frac{2n}{3} + 2\}_{ht}$; red in~\cref{fig:d4},
    \item $\{(a_{\frac{n}{3} - 2}, b_j), j=\frac{2n}{3} + 2,\dots,\frac{2n}{3} + 4\}_{ht}$; light red in~\cref{fig:d4},
    \item $\{(a_i, b_{\frac{2n}{3} + 4}), i = \frac{n}{3} - 1, \frac{n}{3}\}_{ht}$; orange in~\cref{fig:d4},
    \item $\{(a_i, b_1), i = \frac{n}{3} + 6, \frac{n}{3} + 7\}_{hh}$; dark orange in~\cref{fig:d4},
    \item $\{(a_i, b_2), i = \frac{n}{3} + 3, \dots, \frac{n}{3} + 6\}_{hh}$; light orange in~\cref{fig:d4},
    \item $\{(a_{\frac{n}{3} + 3}, b_4)\}_{hh}$; dark green in~\cref{fig:d4},
    \item $\{(a_i, b_5),i= \frac{n}{3}+1, \dots, \frac{n}{3} + 3\}_{hh}$; pink in~\cref{fig:d4},
    \item $\{(a_{\frac{n}{3} + 7}, b_j, j = \frac{2n}{3} + 4, \dots, n\}_{ht}$; yellow in~\cref{fig:d4},
    \item $\{(a_{\frac{n}{3} + 8}, b_j), j = \frac{2n}{3} + 5, \dots, n\}_{hh}$; dark blue in~\cref{fig:d4},
    \item $\{(a_i, b_{\frac{2n}{3} + 6}), i = \frac{n}{3} + 8, \dots, n\}_{hh}$; blue in~\cref{fig:d4},
    \item $\{(a_i, b_{\frac{2n}{3} - 7}), i =1, \dots, \frac{n}{3}\}_{tt}$; green in~\cref{fig:d4},
    \item $\{(a_{\frac{n}{3}}, b_j), j = 1, \dots, \frac{2n}{3}-6 \}_{tt}$; dark pink in~\cref{fig:d4},
    \item $\{(a_{n-1}, b_j), j = \frac{2n}{3} -5, \frac{2n}{3} -4\}_{tt}$; black in~\cref{fig:d4},
    \item $\{(a_i, b_{\frac{2n}{3}-4}),i=n-2,n-3\}_{tt}$; gray in~\cref{fig:d4},
    \item $\{(a_{n-3}, b_j), j = \frac{2n}{3} - 3, \frac{2n}{3} - 2\}_{tt}$; light blue in~\cref{fig:d4},
    \item $\{(a_i, b_{\frac{2n}{3}-1}),i=\frac{n}{3}+1,\dots,n-3\}_{tt}$; dark gray in~\cref{fig:d4}.
\end{itemize}

\begin{figure*}[ht]
	\centering	
    \includegraphics[page=4,width=\textwidth]{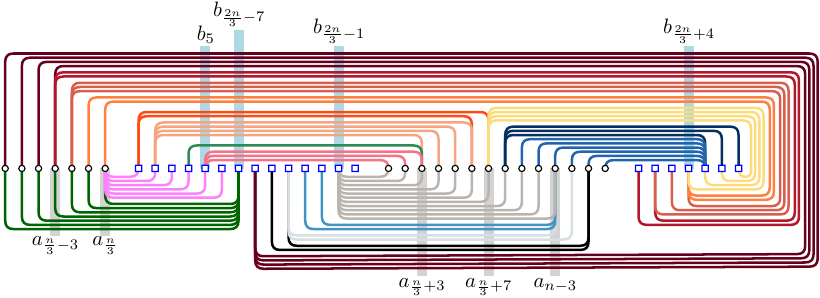}
    \caption{Page $\frac{n}{3}-3$ of $\mathcal{L}$.}
 \label{fig:d4}
\end{figure*}

\clearpage

\smallskip\noindent Page $\frac{n}{3} - 2$ of $\mathcal{L}$ contains the following \textcolor{blue}{$\frac{10n}{3}+13$} edges:
 
\begin{itemize}[-]
\setlength\itemsep{0em}
    \item $\{(a_i, b_{\frac{2n}{3}-4}), i = 1, \dots, \frac{n}{3}-2\}_{ht}$; dark red in~\cref{fig:dn2},
    \item $\{(a_{\frac{n}{3} - 2}, b_{\frac{2n}{3} + 1})\}_{ht}$; red in~\cref{fig:dn2},
    \item $\{(a_{\frac{n}{3} - 1}, b_j), j= \frac{2n}{3} + 1,\dots,\frac{2n}{3} + 3 \}_{ht}$; light red in~\cref{fig:dn2},
    \item $\{(a_{\frac{n}{3}}, b_{\frac{2n}{3} + 3}\}_{ht}$; dark gray in~\cref{fig:dn2},
    \item $\{(a_i, b_1), i = \frac{n}{3} + 2,\dots, \frac{n}{3} + 5\}_{hh}$; light pink in~\cref{fig:dn2},
    \item $ \{(a_{\frac{n}{3} + 2}, b_2), (a_{\frac{n}{3} + 2}, b_4), (a_{\frac{n}{3} + 1}, b_4)\}_{hh}$; gold in~\cref{fig:dn2},
    \item $\{(a_{\frac{n}{3} + 5}, b_j), j = \frac{2n}{3} + 3, \dots, n\}_{ht}$; gray in~\cref{fig:dn2}
    \item $\{(a_{\frac{n}{3} + 6}, b_j), j = \frac{2n}{3} + 3, \dots, n\}_{hh}$; light blue in~\cref{fig:dn2},
    \item $\{(a_{i}, b_{\frac{2n}{3}+3}), i = \frac{2n}{3} + 6, \dots, n\}_{hh}$; blue in~\cref{fig:dn2},
    \item $\{(a_{i}, b_{\frac{2n}{3}-5}), i = 1, \dots, \frac{n}{3}-2\}_{tt}$; orange in~\cref{fig:dn2},
    \item $\{(a_{\frac{n}{3}-2}, b_j), j = 1, \dots, \frac{2n}{3}-4\}_{tt}$; dark pink in~\cref{fig:dn2},
    \item $\{(a_{n - 1}, b_{\frac{2n}{3} - 3})\}_{tt}$; black in~\cref{fig:dn2},
    \item $\{(a_{n - 2}, b_j), j = \frac{2n}{3} - 3, \dots, \frac{2n}{3}\}_{tt}$; pink in~\cref{fig:dn2},
    \item $\{(a_i, b_{\frac{2n}{3}}), i = \frac{n}{3} + 1, \dots, n - 3\}_{tt}$; dark blue in~\cref{fig:dn2}.
\end{itemize}

\begin{figure*}[ h!]
	\centering	
    \includegraphics[page=3,width=\textwidth]{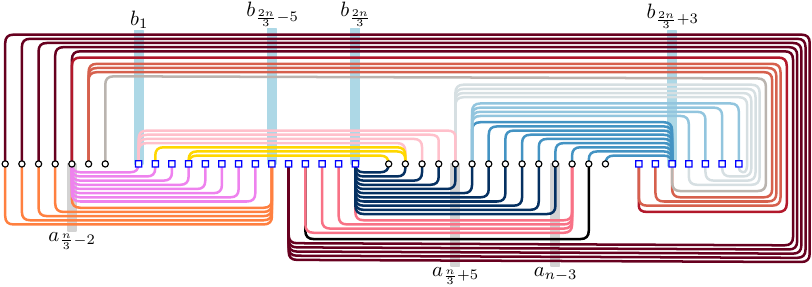}
    \caption{Page $\frac{n}{3}-2$ of $\mathcal{L}$.}
 \label{fig:dn2}
\end{figure*}

\clearpage

\smallskip\noindent Page $\frac{n}{3} - 1$ of $\mathcal{L}$ contains the following \textcolor{blue} {$\frac{11n}{3}-8$} edges:   
 
\begin{itemize}[-]
\setlength\itemsep{0em}
    \item $\{(a_i, b_{\frac{2n}{3}-3}), i = 1, \dots, \frac{n}{3}-1\}_{hh}$; dark red in~\cref{fig:dn1},
    \item $\{(a_{\frac{n}{3}-1}, b_j), j = 1, \dots, \frac{2n}{3}-4\}_{hh}$; light red in~\cref{fig:dn1},
    \item $\{(a_i, b_{\frac{2n}{3}-2}), i = 1, \dots, \frac{n}{3}-1\}_{th}$; red in~\cref{fig:dn1},
    \item $\{(a_{n - 1}, b_j), j = \frac{2n}{3} - 2, \dots, \frac{2n}{3}\}_{ht}$; light orange in~\cref{fig:dn1},
    \item $\{(a_{\frac{n}{3} + 2}, b_{\frac{2n}{3} + 1})\}_{ht}$; dark blue in~\cref{fig:dn1},
    \item $\{(a_{\frac{n}{3} + 3}, b_j), j = \frac{2n}{3} + 1, \dots, n\}_{ht}$; gray in~\cref{fig:dn1},
    \item $\{(a_{\frac{n}{3} + 4}, b_j), j = \frac{2n}{3} + 3, \dots, n\}_{hh}$; light blue in~\cref{fig:dn1},
    \item $\{(a_{\frac{n}{3} + 5}, b_{\frac{2n}{3} + 2}), (a_{\frac{n}{3} + 5}, b_{\frac{2n}{3} + 3})\}_{hh}$; green in~\cref{fig:dn1},
    \item $\{(a_i, b_{\frac{2n}{3} + 2}), i = \frac{n}{3} + 6, \dots, n\}_{hh}$; blue in~\cref{fig:dn1},
    \item $\{(a_{n - 1}, b_1), (a_{n - 2}, b_1), (a_{n - 2}, b_2), (a_{n - 3}, b_2)\}_{tt}$; orange in~\cref{fig:dn1},
    \item $\{(a_i, b_j), (a_{i - 1}, b_j),  i = n - 3, \dots, \frac{2n}{3} +2, j = 4, \dots, \frac{n}{3}-2\}_{tt}$; yellow in~\cref{fig:dn1},
    \item $\{(a_i,b_{\frac{n}{3}-1},i = \frac{2n}{3}+2,\dots,\frac{2n}{3}\}_{tt}$; pink in~\cref{fig:dn1},
    \item $\{(a_{\frac{2n}{3}}, b_j), j = \frac{n}{3}, \frac{n}{3} + 1\}_{tt}$; black in~\cref{fig:dn1},
    \item $\{(a_i, b_{\frac{n}{3} + 2}), i = \frac{n}{3}+1, \dots, \frac{2n}{3}\}_{tt}$; dark pink in~\cref{fig:dn1},
\end{itemize}
 
\begin{figure*}[ h!]
	\centering	
    \includegraphics[page=2,width=\textwidth]{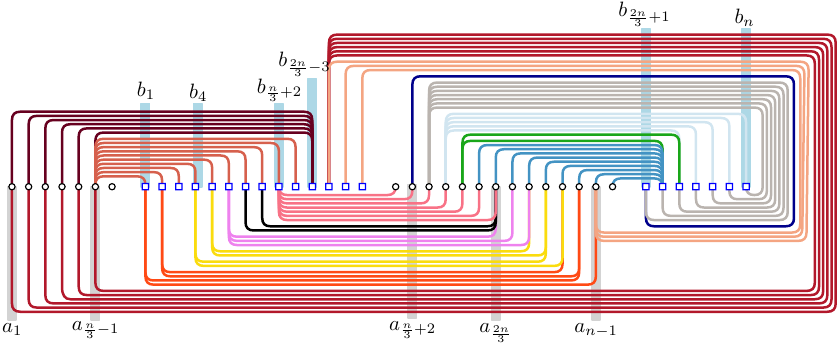}
    \caption{Page $\frac{n}{3}-1$ of $\mathcal{L}$.}
 \label{fig:dn1}
\end{figure*}

\clearpage

\smallskip\noindent  Page $\frac{n}{3}$ of $\mathcal{L}$ contains the following \textcolor{blue} {$4n+14$} edges:
 
\begin{itemize}[-]
\setlength\itemsep{0em}
   \item $\{(a_{i}, b_{\frac{2n}{3}}), i = 1, \dots, \frac{n}{3} \}_{ht}$; dark red in~\cref{fig:dn}, 
   \item $\{(a_{\frac{n}{3} + 1}, b_j), j = 1, 2, 3\}_{ht}$; red in~\cref{fig:dn}, 
    \item $\{(a_i, b_3), i = \frac{n}{3} + 1, \dots, n\}_{ht}$; dark orange in~\cref{fig:dn},
    \item $\{(a_n, b_j), j = 4, \dots, \frac{2n}{3}\}_{ht}$; orange in~\cref{fig:dn},
    \item $\{(a_{\frac{n}{3} + 2}, b_j), j = \frac{2n}{3} + 2, \dots, n\}_{hh}$; gray in~\cref{fig:dn},
    \item $\{(a_{\frac{n}{3} + 4}, b_{\frac{2n}{3} + 2})\}_{hh}$; yellow in~\cref{fig:dn},
    \item $\{(a_i, b_{\frac{2n}{3} + 1}), i = \frac{n}{3} + 4, \dots, n\}_{hh}$; light blue in~\cref{fig:dn},
    \item $\{(a_{\frac{n}{3} + 1}, b_j), j = \frac{2n}{3} + 1, \dots, n\}_{ht}$; light orange in~\cref{fig:dn},
    \item $\{(a_i, b_{\frac{2n}{3} - 1}), i = 1, \dots, \frac{n}{3}\}_{ht}$; dark blue in~\cref{fig:dn},
    \item $\{(a_{\frac{n}{3}}, b_j), j = 1, \dots,\frac{2n}{3} - 1\}_{ht}$; blue in~\cref{fig:dn}.
\end{itemize}
 
\begin{figure*}[ h!]
	\centering	
    \includegraphics[page=1,width=\textwidth]{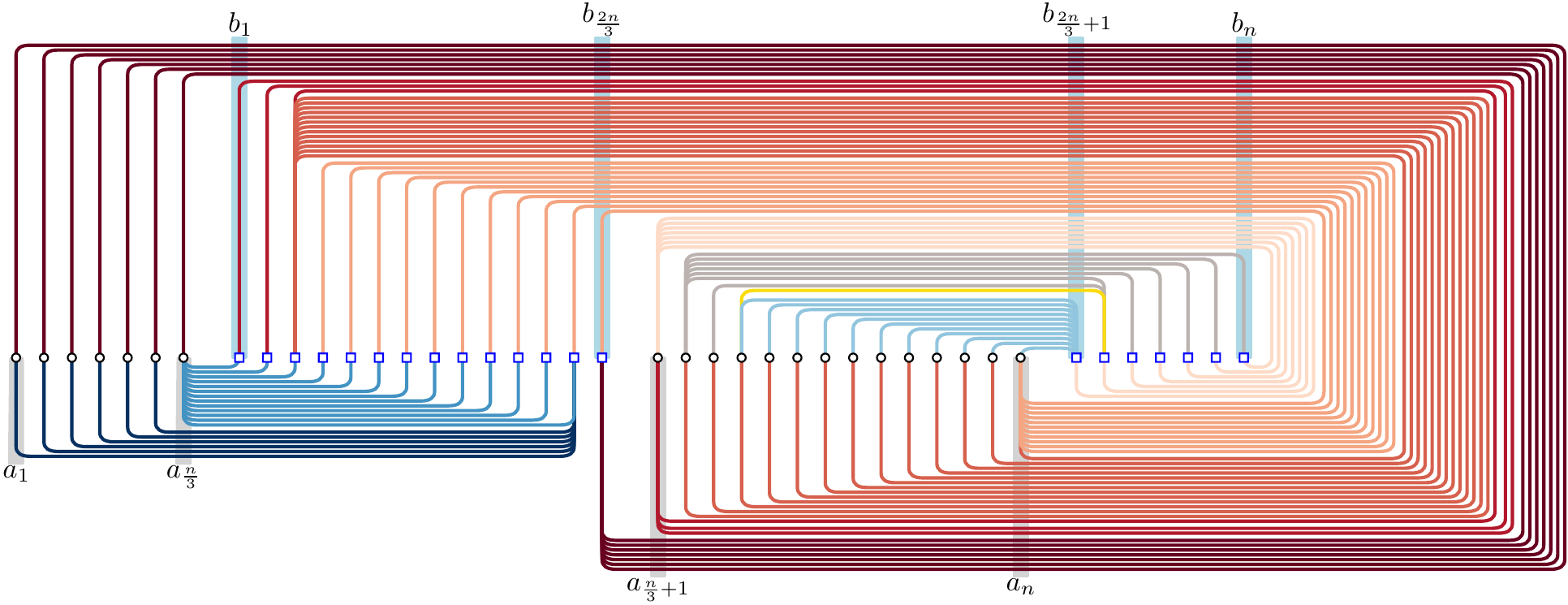}
    \caption{Page $\frac{n}{3}$ of $\mathcal{L}$.}
 \label{fig:dn}
\end{figure*}

\noindent So, in total $\mathcal{L}$ has $2n+16+\frac{8n}{3}+5+\sum_{p=3}^{4}(\frac{8n}{3}+p+5)+\sum_{p=\frac{n}{3}-7}^{\frac{n}{3}-4}(\frac{8n}{3}-2p-1)+\sum_{p=5}^{\frac{n}{3}-8}(\frac{8n}{3}+2p+18)+\frac{10n}{3}+11+\frac{10n}{3}+13+\frac{11n}{3}-8+4n+14 = n^2$ edges. Since no two edges have been assigned to the same deque and all edges in the same deque form a cylindric layout, it follows that the deque number of $K_{n,n}$ is at most $\frac{n}{3}$.
\end{proof}

\begin{figure*}[p]
    \centering	    
    \includegraphics[width=\textwidth,page=2]{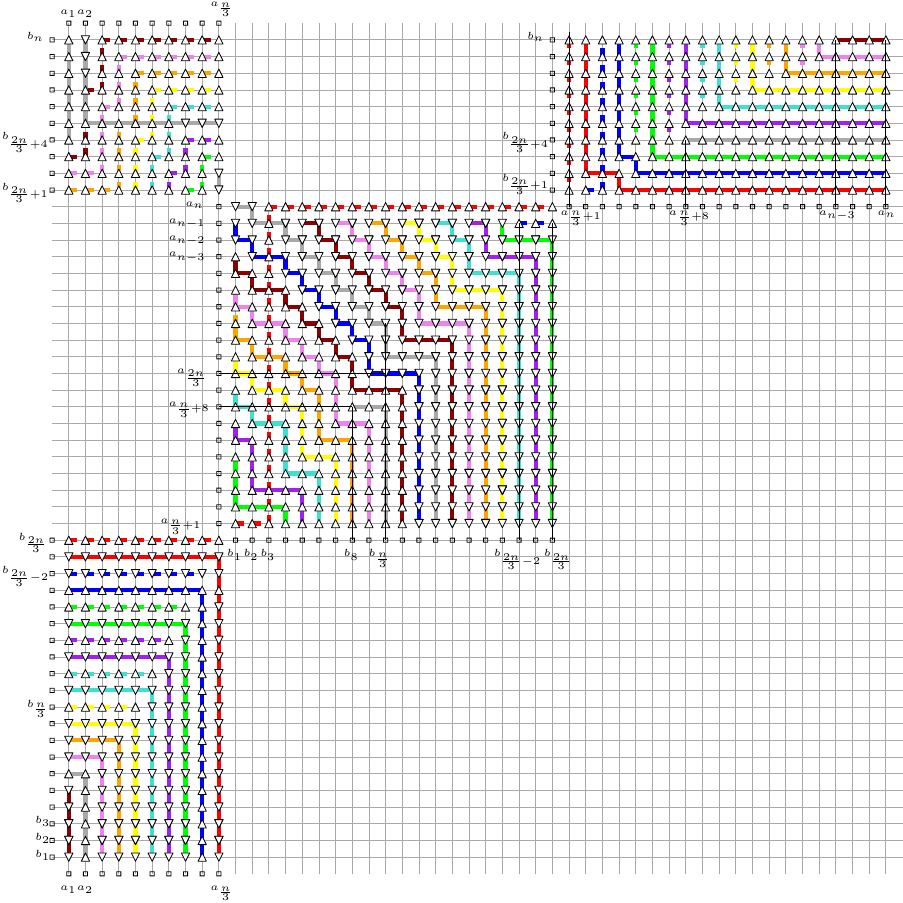}
    \caption{Illustration of the grid representation of a deque layout of $K_{n,n}$, in which paths of the same color correspond to the same deque. The points of the grid that are represented by a triangle pointing up (down) correspond to  head-$\star$ (tail-$\star$, resp.) edges; grid points covered by a solid line are head-head or tail-tail edges; the dashed covered ones are head-tail or tail-head}
    \label{fig:bipartite-matrix-deque}
\end{figure*}

\clearpage

\noindent We complete this section by providing our upper bound on the rique-number of $K_{n,n}$. 

\newcommand{\riquebipartite}{The rique-number of the complete bipartite graph $K_{n,n}$ is at most $\lfloor\frac{n-1}{2}\rfloor-1$.}
\begin{theorem}\label{thm:rique-bipartite}
\riquebipartite
\end{theorem}
    
\begin{proof}

For the proof, we distinguish two cases depending on whether $n$ is odd or even.

\medskip\noindent\textbf{Case 1:} $n$ is odd. We first assume that $n$ is odd and we prove that $K_{n,n}$ admits a rique layout $\mathcal{L}$ with $\lfloor \frac{n}{2} \rfloor - 1$ riques. Let $A=\{a_1,\dots,a_n\}$ and $B=\{b_1,\dots,b_n\}$ be the two parts of $K_{n,n}$, such that $a_1 \prec b_1 \prec a_2 \prec b_2 \prec \dots \prec a_{\lfloor \frac{n}{2} \rfloor} \prec b_{\lfloor \frac{n}{2} \rfloor} \prec b_{\lceil \frac{n}{2} \rceil} \prec \dots \prec b_n \prec a_{\lceil \frac{n}{2} \rceil} \prec \dots \prec a_n$ holds in $\mathcal{L}$. In $\mathcal{L}$, there exist 12 special riques, in particular, the ones in $\{1, 2, 3, 4, 5,6,7, \lfloor \frac{n}{2} \rfloor-5, \lfloor \frac{n}{2} \rfloor-4, \lfloor \frac{n}{2} \rfloor-3, \lfloor \frac{n}{2} \rfloor-2, \lfloor \frac{n}{2} \rfloor-1\}$; see Fig.~\ref{fig:k29,29completebipartite}.   

\smallskip\noindent Page~1 of $\mathcal{L}$ contains the following \textcolor{blue}{3n} edges:
\begin{itemize}[-]
\setlength\itemsep{0em}
\item $\{(a_1,b_j), j=1,\dots, n\}_{ht}$; dark red in~\cref{fig:o1}, 
\item $\{(a_i,b_1), i = \lceil \frac{n}{2} \rceil ,\dots,n\}_{ht}$; red in~\cref{fig:o1},
\item $\{(a_{\lfloor \frac{n}{2} \rfloor},b_j), j = 2 ,\dots, \lfloor \frac{n}{2} \rfloor\}_{hh}$; gray in~\cref{fig:o1},
\item $\{(a_i,b_2), i = \lceil \frac{n}{2} \rceil ,\dots,n\}_{hh}$; blue in~\cref{fig:o1},
\item $\{(a_{\lceil \frac{n}{2} \rceil},b_j), j =  \lceil \frac{n}{2} \rceil ,\dots, n \}_{hh}$; light blue in~\cref{fig:o1}.
\end{itemize}
 
\begin{figure*}[h!]
	\centering
	\includegraphics[page=1]{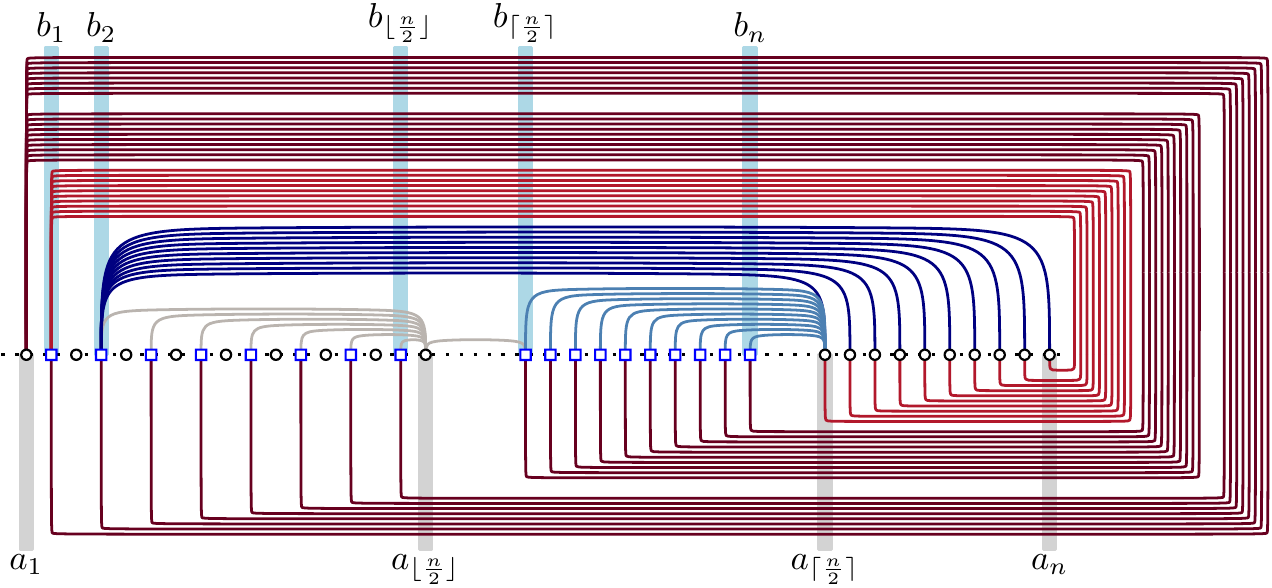}
    \caption{Page $1$ of $\mathcal{L}$ when $n$ is odd.}
 \label{fig:o1}
\end{figure*}
 
\smallskip\noindent Page~2 of $\mathcal{L}$ contains the following \textcolor{blue}{$\frac{5n-1}{2} + 1 $} edges:
\begin{itemize}[-]
\setlength\itemsep{0em}
\item $\{(a_i,b_1), i = 2 ,\dots,\lfloor \frac{n}{2} \rfloor\}_{ht}$; dark red in~\cref{fig:o2},
\item $\{(a_2,b_{\lfloor \frac{n}{2} \rfloor}) \}_{ht}$; yellow in~\cref{fig:o2}, 
\item $\{(a_2,b_j), j= \lceil \frac{n}{2} \rceil ,\dots, n\}_{ht}$; light red in~\cref{fig:o2}, 
\item $\{(a_3,b_n) \}_{ht}$; red in~\cref{fig:o2}, 
\item $\{(a_i,b_3), i = \lceil \frac{n}{2} \rceil ,\dots,n\}_{ht}$; dark blue in~\cref{fig:o2},
\item $\{(a_{n-3},b_j), j = \lfloor \frac{n}{2} \rfloor - 2, \dots, \lceil \frac{n}{2} \rceil + 1\}_{hh}$; light blue in~\cref{fig:o2},
\item $\{(a_i,b_{\lceil \frac{n}{2} \rceil + 1}), i = \lceil \frac{n}{2} \rceil + 1 , \dots, n-4\}_{hh}$; blue in~\cref{fig:o2},
\item $\{(a_{\lceil \frac{n}{2} \rceil + 1 },b_j), j = \lceil \frac{n}{2} \rceil + 2 ,\dots,n\}_{hh}$; gray in~\cref{fig:o2}.
\end{itemize}
 
\begin{figure*}[h!]
	\centering
	\includegraphics[page=2]{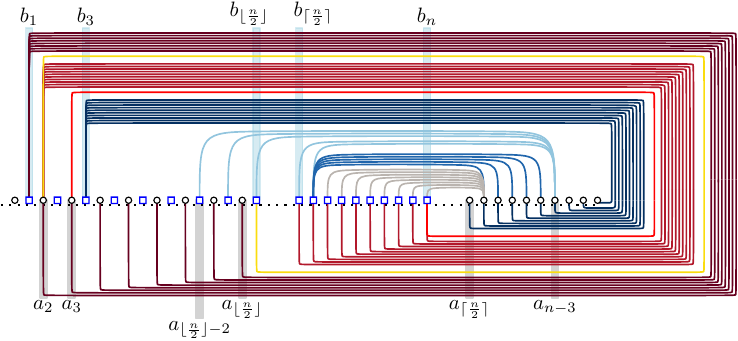}
    \caption{Page $2$ of $\mathcal{L}$ when $n$ is odd.}
 \label{fig:o2}
\end{figure*}
 
\smallskip\noindent For $p=3,4,5$, page $p$ of $\mathcal{L}$ contains the following \textcolor{blue}{$\left(\frac{5n-1}{2}\right) + 1 $} edges
\begin{itemize}[-]
\setlength\itemsep{0em}
\item $\{(a_{p-1},b_j), j= 2,\dots, \lfloor \frac{n}{2} \rfloor-p+2\}_{ht}$; dark red in~\cref{fig:o345},
\item $\{(a_p,b_j), j=\lfloor \frac{n}{2} \rfloor-p+2,\dots, n-p+2\}_{ht}$; red in~\cref{fig:o345}, 
\item $\{(a_{p+1},b_j), j= n-p+2, \dots, n\}_{ht}$; light red in~\cref{fig:o345},
\item $\{(a_i,b_{p+1}), i=\lceil \frac{n}{2} \rceil,\dots, n\}_{ht}$; dark blue in~\cref{fig:o345},
\item $\{(a_{n+(p-5)},b_j), j = \lfloor \frac{n}{2} \rfloor - 2, \dots, \lceil \frac{n}{2} \rceil + (p-1)\}_{hh}$; gray in~\cref{fig:o345},
\item $\{(a_i,b_{\lceil \frac{n}{2} \rceil + (p-1)}), i = \lceil \frac{n}{2} \rceil + (p-1) , \dots, n+(p-6)\}_{hh}$; blue in~\cref{fig:o345},
\item $\{(a_{\lceil \frac{n}{2} \rceil + (p-1) },b_j), j = \lceil \frac{n}{2} \rceil + p ,\dots,n\}_{hh}$; light blue in~\cref{fig:o345}.
\end{itemize}
 
\begin{figure*}[h!]
	\centering
	\includegraphics[page=3]{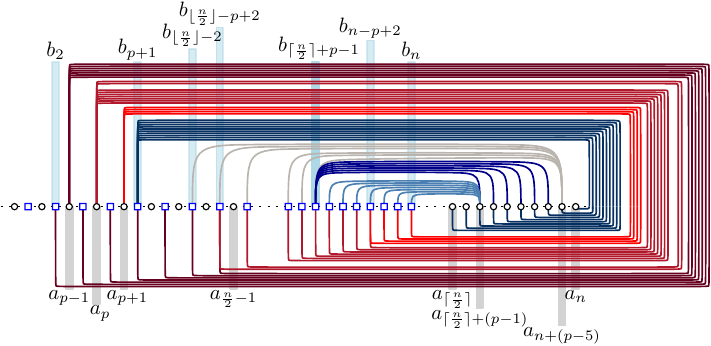}
    \caption{Page $p=3,4,5$ of $\mathcal{L}$ when $n$ is odd.}
 \label{fig:o345}
\end{figure*}

\clearpage

\smallskip\noindent Page 6 of $\mathcal{L}$ contains the following \textcolor{blue}{$\frac{5n-1}{2}-8$} edges:
\begin{itemize}[-]
\setlength\itemsep{0em}
\item $\{(a_5,b_j), j= 2,\dots, \lfloor \frac{n}{2} \rfloor-4\}_{ht}$; dark red in~\cref{fig:o6},
\item $\{(a_6,b_j), j=\lfloor \frac{n}{2} \rfloor-4,\dots, n-4\}_{ht}$; red in~\cref{fig:o6}, 
\item $\{(a_7,b_j), j= n-4, \dots, n\}_{ht}$; light red in~\cref{fig:o6},
\item $\{(a_i,b_7), i=\lceil \frac{n}{2} \rceil,\dots, n\}_{ht}$; dark blue in~\cref{fig:o6},
\item $\{(a_i,b_{\lceil \frac{n}{2} \rceil + 5}), i = \lceil \frac{n}{2} \rceil +5 , \dots,n)\}_{hh}$; blue in~\cref{fig:o6},
\item $\{(a_{\lceil \frac{n}{2} \rceil + 5 },b_j), j = \lceil \frac{n}{2} \rceil + 6,\dots,n\}_{hh}$; light red in~\cref{fig:o6}.
\end{itemize}
 
\begin{figure*}[h!]
	\centering
	\includegraphics[page=4]{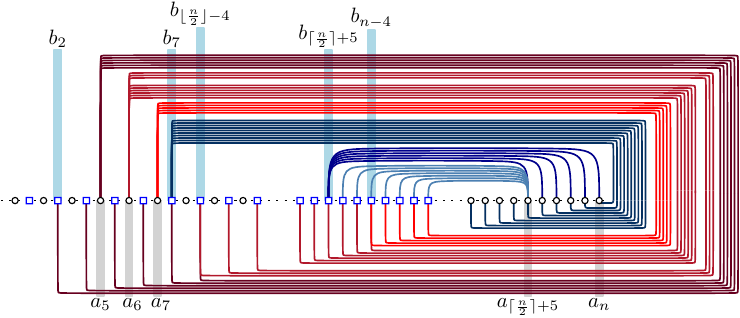}
    \caption{Page $6$ of $\mathcal{L}$ when $n$ is odd.}
 \label{fig:o6}
\end{figure*}

\begin{figure*}[b]
	\centering
	\includegraphics[page=5]{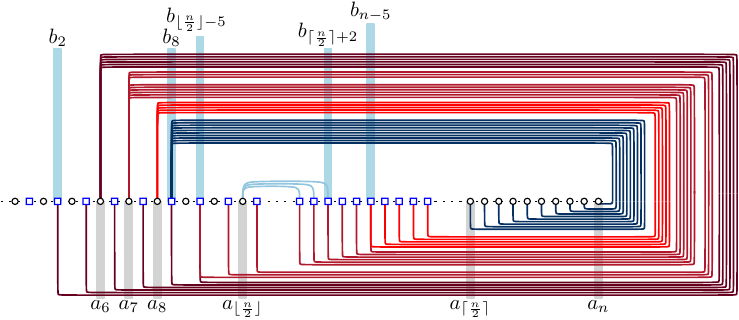}
    \caption{Page $7$ of $\mathcal{L}$ when $n$ is odd.}
 \label{fig:o7}
\end{figure*}
 
\smallskip\noindent Page 7 of $\mathcal{L}$ contains the following \textcolor{blue}{$\frac{3n+1}{2} +4$} edges:
\begin{itemize}[-]
\setlength\itemsep{0em}
\item $\{(a_6,b_j), j= 2,\dots, \lfloor \frac{n}{2} \rfloor-5\}_{ht}$; dark red in~\cref{fig:o7},
\item $\{(a_7,b_j), j=\lfloor \frac{n}{2} \rfloor-5,\dots, n-5\}_{ht}$; red in~\cref{fig:o7}, 
\item $\{(a_8,b_j), j= n- 5, \dots, n\}_{ht}$; light red in~\cref{fig:o7},
\item $\{(a_i,b_8), i=\lceil \frac{n}{2} \rceil,\dots, n\}_{ht}$; blue in~\cref{fig:o7},
\item $\{(a_{\lfloor \frac{n}{2} \rfloor},b_j), j = \lceil \frac{n}{2} \rceil , \dots,  \lceil \frac{n}{2} \rceil+2)\}_{hh}$; light blue in~\cref{fig:o7}.
\end{itemize}

\smallskip\noindent For $p =8, \dots, \frac{n-1}{2}-6 $, page $p$ of $\mathcal{L}$ contains the following \textcolor{blue}{$\frac{5n+3}{2} -2p + 4$} edges:
\begin{itemize}[-]
\setlength\itemsep{0em}
\item $\{(a_{p-1},b_j), j= 2,\dots, \lfloor \frac{n}{2} \rfloor-p+2\}_{ht}$; dark red in~\cref{fig:op},
\item $\{(a_p,b_j), j=\lfloor \frac{n}{2} \rfloor-p+2,\dots, n-p+2\}_{ht}$; red in~\cref{fig:op}, 
\item $\{(a_{p+1},b_j), j=n-p+2, \dots, n\}_{ht}$; light red in~\cref{fig:op},
\item $\{(a_i,b_{p+1}), i=\lceil \frac{n}{2} \rceil,\dots, n\}_{ht}$; dark blue in~\cref{fig:op},
\item $\{(a_i,b_{\lceil \frac{n}{2} \rceil+p-2}), i = \lceil \frac{n}{2} \rceil +p-2, \dots,n)\}_{hh}$; blue in~\cref{fig:op},
\item $\{(a_{\lceil \frac{n}{2} \rceil+p-2},b_j), j = \lceil \frac{n}{2} \rceil + p-1 ,\dots,n\}_{hh}$; light blue in~\cref{fig:op}.
\end{itemize}
 
\begin{figure*}[h!]
	\centering
	\includegraphics[page=6]{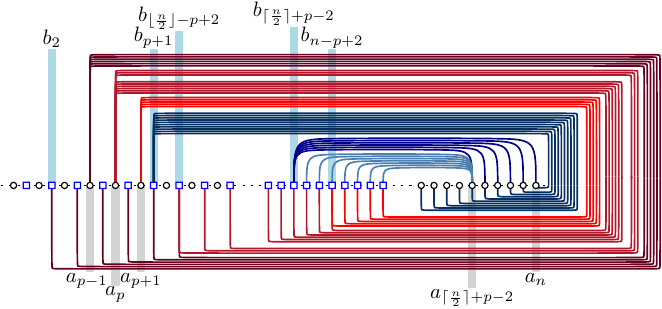}
    \caption{Page $p =8, \dots, \frac{n-1}{2}-6 $ of $\mathcal{L}$ when $n$ is odd.}
 \label{fig:op}
\end{figure*}

\clearpage

\smallskip\noindent Page  $\frac{n-1}{2} - 5$ of $\mathcal{L}$ contains the following \textcolor{blue}{2n-3} edges:
\begin{itemize}[-]
\setlength\itemsep{0em}
\item $\{(a_{\lfloor \frac{n}{2} \rfloor-6},b_j), j= 2,\dots, \lfloor \frac{n}{2} \rfloor -7\}_{ht}$; dark red in~\cref{fig:o},
\item $\{(a_{(\lfloor \frac{n}{2} \rfloor - 5)},b_j), j=\lfloor \frac{n}{2} \rfloor -7,\dots, n-7\}_{ht}$; red in~\cref{fig:o}, 
\item $\{(a_{\lfloor \frac{n}{2} \rfloor - 4},b_j), j= n-7, \dots, n\}_{ht}$; light red in~\cref{fig:o},
\item $\{(a_i,b_{\lfloor \frac{n}{2} \rfloor -4}), i=\lceil \frac{n}{2} \rceil,\dots, n\}_{ht}$; light blue in~\cref{fig:o},
\item $\{(a_i,b_{\lfloor \frac{n}{2} \rfloor - 3}), i = \lceil \frac{n}{2} \rceil + 5 , \dots,n)\}_{hh}$; blue in~\cref{fig:o}.
\end{itemize}

\begin{figure*}[h!]
	\centering
	\includegraphics[page=7]{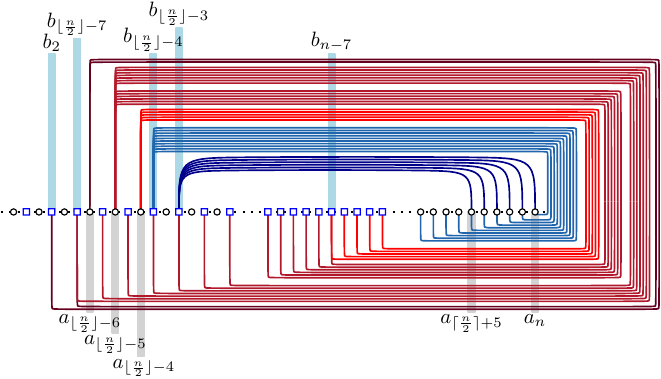}
    \caption{Page $\frac{n-1}{2} - 5$ of $\mathcal{L}$ when $n$ is odd.}
 \label{fig:o}
\end{figure*}

\smallskip\noindent For $p=\frac{n-1}{2} - k,$ with $ k \in \{4,3,2\}$, page $p$ of $\mathcal{L}$ contains the following \textcolor{blue}{$\frac{3n-1}{2}-2k+20$} edges:
 
\begin{itemize}[-]
\setlength\itemsep{0em}
\item $\{(a_{p-1},b_j), j= 2,\dots, \lfloor \frac{n}{2} \rfloor-p-4\}_{ht}$; dark red in~\cref{fig:ok},
\item $\{(a_p,b_j), j=\lfloor \frac{n}{2} \rfloor-p-4,\dots, n-p-4\}_{ht}$; red in~\cref{fig:ok}, 
\item $\{(a_{p+1},b_j), j= n-p-4, \dots, n\}_{ht}$; light red in~\cref{fig:ok},
\item $\{(a_i,b_{p+1}), i=\lceil \frac{n}{2} \rceil,\dots, \lceil \frac{n}{2} \rceil + k \}_{ht}$; dark pink in~\cref{fig:ok},
\item $\{(a_i,b_{p+2}), i= \lceil \frac{n}{2} \rceil + k,\dots, n-8+k \}_{ht}$; light pink in~\cref{fig:ok},
\item $\{(a_{n-8+k},b_j), j = \lfloor \frac{n}{2} \rfloor - 2+k, \dots, \lceil \frac{n}{2} \rceil\}_{ht}$; pink in~\cref{fig:ok},
\item $\{(a_i,b_{n+(2k-9)}), i = n+(k-8), \dots, n \}_{ht}$; dark blue in~\cref{fig:ok},
\item $\{(a_i,b_{n+(2k-8)}), i = n+(k-8) , \dots,n)\}_{hh}$; blue in~\cref{fig:ok},
\item $\{(a_{n+(k-8)},b_j), j = n+(2k-7) , \dots,n)\}_{hh}$; gray in~\cref{fig:ok}.
\end{itemize}
 
\begin{figure*}[h!]
	\centering
	\includegraphics[page=8]{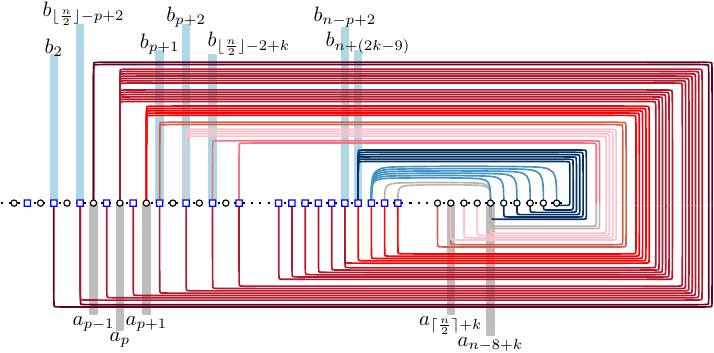}
    \caption{Page $p=\frac{n-1}{2} - k,$ with $ k \in \{4,3,2\}$ of $\mathcal{L}$ when $n$ is odd.}
 \label{fig:ok}
\end{figure*}
 
\smallskip\noindent Page $\frac{n-1}{2} - 1$ of $\mathcal{L}$ contains the following \textcolor{blue}{$\frac{3n+1}{2}+18$} edges:
 
\begin{itemize}[-]
\setlength\itemsep{0em}
\item $\{(a_{\lfloor \frac{n}{2} \rfloor -1},b_j), j= 2,\dots, \lceil \frac{n}{2} \rceil+3\}_{ht}$; dark red in~\cref{fig:o9},
\item $\{(a_{\lfloor \frac{n}{2} \rfloor},b_j), j=\lceil \frac{n}{2} \rceil+3,\dots, n\}_{ht}$; red in~\cref{fig:o9}, 
\item $\{(a_i,b_{\lfloor \frac{n}{2} \rfloor}), i=\lceil \frac{n}{2} \rceil,\lceil \frac{n}{2} \rceil+1\}_{ht}$; light red in~\cref{fig:o9},
\item $\{(a_i,b_{\lceil \frac{n}{2} \rceil}), i=\lceil \frac{n}{2} \rceil + 1, \dots,n-7\}_{ht}$; dark pink in~\cref{fig:o9},
\item $\{(a_i,b_{n-7}), i=\lceil \frac{n}{4} \rceil + 1, \dots, n\}_{ht}$; pink in~\cref{fig:o9},
\item $\{(a_i,b_{n-6}), i=\lfloor \frac{n}{4} \rfloor + 1 , \dots, n\}_{hh}$; dark blue in~\cref{fig:o9},
\item $\{(a_{\lceil \frac{n}{4} \rceil + 1},b_j), j=n-5, \dots, n\}_{hh}$; blue in~\cref{fig:o9},
\item $\{(a_{\lfloor \frac{n}{2} \rfloor - 2},b_j), j=2,3\}_{hh}$; gray in~\cref{fig:o9}.
\end{itemize}
 
\begin{figure*}[h!]
	\centering
	\includegraphics[page=9]{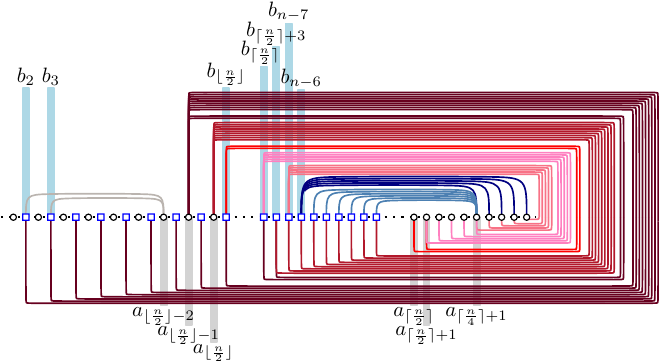}
    \caption{Page $\frac{n-1}{2} - 1$ of $\mathcal{L}$ when $n$ is odd.}
 \label{fig:o9}
\end{figure*}
 
\noindent So, in total $\mathcal{L}$ has $3n+\frac{5n-1}{2} + 1+3\left(\frac{5n-1}{2} + 1\right)+\frac{5n-1}{2}-8+\frac{3n+1}{2}+4+\sum_{p=8}^{\frac{n-1}{2} -6}\left(\frac{5n+3}{2}-2p+4\right)+2n-3+\sum_{k=2}^{4}(\frac{3n-1}{2}-2k+20)+(\frac{3n+1}{2}+18)$ = $n^2$ edges. Since no two edges in the same rique deviate from the properties of cylindric layouts, it follows that the rique number of $k_{n,n}$ is at most $\lfloor \frac{n-1}{2} \rfloor -1 $ when n is odd. 

\medskip\noindent\textbf{Case 2:} $n$ is even. When $n$ is even, we prove that $K_{n,n}$ admits a rique layout $\mathcal{L}$ with $\frac{n}{2}-2$ riques.
Let $A=\{a_1,\dots,a_n\}$ and $B=\{b_1,\dots,b_n\}$ be the two parts of $K_{n,n}$, such that $a_1 \prec b_1 \prec a_2 \prec b_2 \prec \dots \prec a_{\frac{n}{2}-1} \prec b_{\frac{n}{2}-1} \prec b_{\frac{n}{2} } \prec \dots \prec b_n \prec a_{\frac{n}{2}} \prec \dots \prec a_n$ holds in $\mathcal{L}$. In $\mathcal{L}$, there exist 12 special riques, in particular, the ones in $\{1, 2, 3, 4, 5,6,7, \frac{n}{2}-6, \frac{n}{2}-5, \frac{n}{2}-4, \frac{n}{2}-3, \frac{n}{2}-2\}$; see Fig.~\ref{fig:k30,30completebipartite}.  

\smallskip\noindent Page~1 of $\mathcal{L}$ contains the following \textcolor{blue}{$3n+1$} edges:
\begin{itemize}[-]
\setlength\itemsep{0em}
\item $\{(a_1,b_j), j=1,\dots, n\}_{ht}$; dark red in~\cref{fig:e1}, 
\item $\{(a_i,b_1), i = \frac{n}{2},\dots,n\}_{ht}$; red in~\cref{fig:e1},
\item $\{(a_{ \frac{n}{2}-1},b_j), j = 2 ,\dots, \frac{n}{2}-1\}_{hh}$; light blue in~\cref{fig:e1},
\item $\{(a_i,b_2), i = \frac{n}{2} ,\dots,n\}_{hh}$; light red in~\cref{fig:e1},
\item $\{(a_{\frac{n}{2}},b_j), j =   \frac{n}{2} ,\dots, n \}_{hh}$; blue in~\cref{fig:e1}.
\end{itemize}
 
\begin{figure*}[h!]
	\centering
	\includegraphics[page=1]{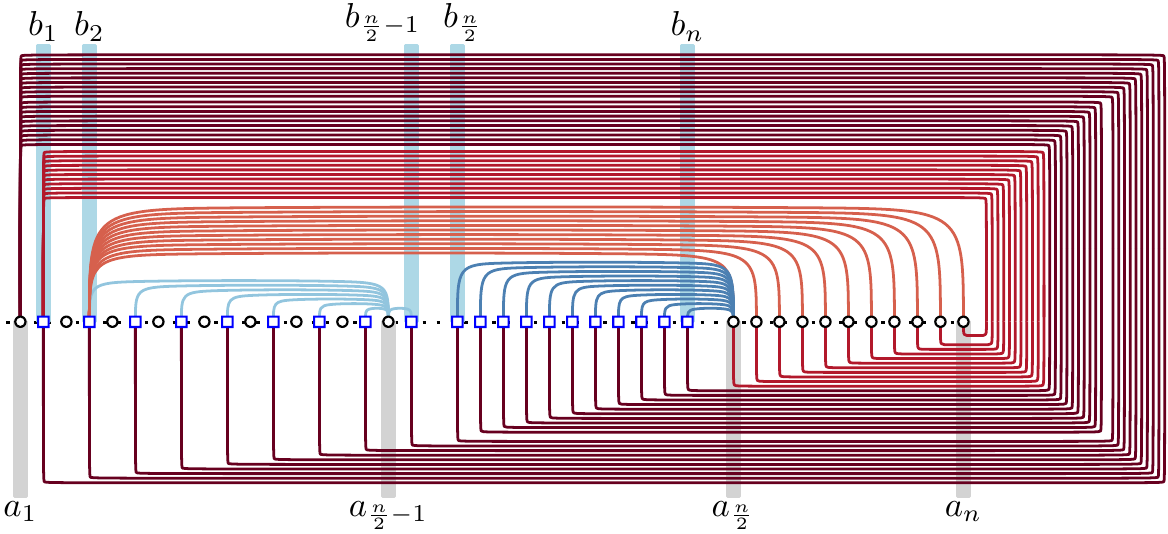}
    \caption{Page $1$ of $\mathcal{L}$ when $n$ is even.}
 \label{fig:e1}
\end{figure*}

\begin{figure*}[b]
	\centering
	\includegraphics[page=2]{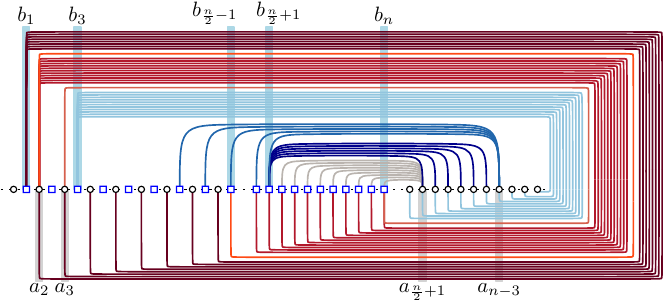}
    \caption{Page $2$ of $\mathcal{L}$ when $n$ is even.}
 \label{fig:e2}
\end{figure*}

\smallskip\noindent Page~2 of $\mathcal{L}$ contains the following \textcolor{blue}{$\frac{5n}{2}+2 $} edges:
\begin{itemize}[-]
\setlength\itemsep{0em}
\item $\{(a_i,b_1), i = 2 ,\dots,\frac{n}{2} -1\}_{ht}$; dark red in~\cref{fig:e2},
\item $\{(a_2,b_{\frac{n}{2} -1}) \}_{ht}$; orange in~\cref{fig:e2}, 
\item $\{(a_2,b_j), j=  \frac{n}{2}  ,\dots, n\}_{ht}$; red in~\cref{fig:e2}, 
\item $\{(a_3,b_n) \}_{ht}$; light red in~\cref{fig:e2}, 
\item $\{(a_i,b_3), i =  \frac{n}{2},\dots,n\}_{ht}$;light blue in~\cref{fig:e2},
\item $\{(a_{n-3},b_j), j =  \frac{n}{2} - 3, \dots, \frac{n}{2} + 1\}_{hh}$; blue in~\cref{fig:e2},
\item $\{(a_i,b_{\frac{n}{2} + 1}), i =  \frac{n}{2} + 1 , \dots, n-4\}_{hh}$; dark blue in~\cref{fig:e2},
\item $\{(a_{\frac{n}{2}+ 1},b_j), j = \frac{n}{2} + 2 ,\dots,n\}_{hh}$; gray in~\cref{fig:e2}.
\end{itemize}

\smallskip\noindent For $p=3,4,5$, page $p$ of $\mathcal{L}$ contains the following \textcolor{blue}{$\frac{5n}{2}+2$} edges
 
\begin{itemize}[-]
\setlength\itemsep{0em}
\item $\{(a_{p-1},b_j), j= 2,\dots, \frac{n}{2}-p+1\}_{ht}$; dark red in~\cref{fig:e3},
\item $\{(a_p,b_j), j=\frac{n}{2}-p+1,\dots,n-p+2\}_{ht}$; red in~\cref{fig:e3}, 
\item $\{(a_{p+1},b_j), j=n-p+2, \dots, n\}_{ht}$; light red in~\cref{fig:e3},
\item $\{(a_i,b_{p+1}), i=\frac{n}{2},\dots, n\}_{ht}$; dark blue in~\cref{fig:e3},
\item $\{(a_{n+p-6},b_j), j = \frac{n}{2} - 3, \dots, \frac{n}{2}+p-1\}_{hh}$; light blue in~\cref{fig:e3},
\item $\{(a_i,b_{\frac{n}{2} + (p-1)}), i = \frac{n}{2} + (p-1) , \dots, n+(p-6)\}_{hh}$; blue in~\cref{fig:e3},
\item $\{(a_{\frac{n}{2}+ (p-1) },b_j), j =\frac{n}{2} + p,\dots,n\}_{hh}$; gray in~\cref{fig:e3}.
\end{itemize}
 
\begin{figure*}[h!]
	\centering
	\includegraphics[page=3]{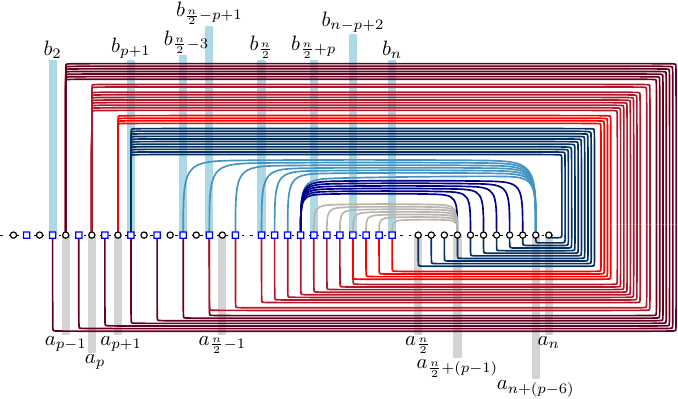}
    \caption{Page $p=3,4,5$ of $\mathcal{L}$ when $n$ is even.}
 \label{fig:e3}
\end{figure*}
 
\smallskip\noindent Page 6 of $\mathcal{L}$ contains the following \textcolor{blue}{$\frac{5n}{2}-7$} edges:
\begin{itemize}[-]
\setlength\itemsep{0em}
\item $\{(a_5,b_j), j= 2,\dots, \frac{n}{2}-5\}_{ht}$; dark red in~\cref{fig:e6},
\item $\{(a_6,b_j), j=\frac{n}{2}-5,\dots, n-4\}_{ht}$; red in~\cref{fig:e6}, 
\item $\{(a_7,b_j), j= n-4, \dots, n\}_{ht}$; light red in~\cref{fig:e6},
\item $\{(a_i,b_7), i=\frac{n}{2},\dots, n\}_{ht}$; dark blue in~\cref{fig:e6},
\item $\{(a_i,b_{\frac{n}{2} + 5}), i =\frac{n}{2}+5 , \dots,n)\}_{hh}$; blue in~\cref{fig:e6},
\item $\{(a_{\frac{n}{2}  + 5 },b_j), j = \frac{n}{2} + 6,\dots,n\}_{hh}$; gray in~\cref{fig:e6}.
\end{itemize}
 
\begin{figure*}[h!]
	\centering
	\includegraphics[page=4]{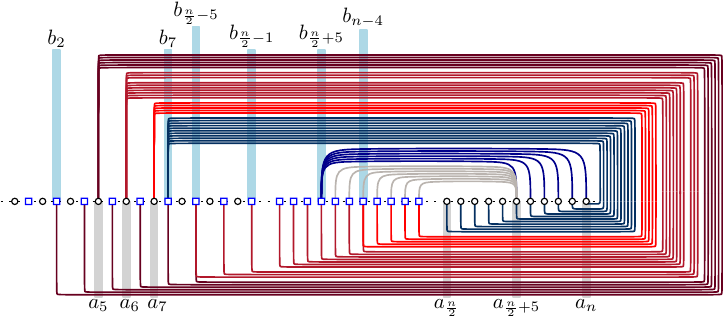}
    \caption{Page $6$ of $\mathcal{L}$ when $n$ is even.}
 \label{fig:e6}
\end{figure*}
 
\smallskip\noindent Page 7 of $\mathcal{L}$ contains the following \textcolor{blue}{$\frac{3n}{2}+23$} edges:
\begin{itemize}[-]
\setlength\itemsep{0em}
\item $\{(a_6,b_j), j= 2,\dots,\frac{n}{2}-6\}_{ht}$; dark red in~\cref{fig:e7},
\item $\{(a_7,b_j), j=\frac{n}{2}-6,\dots, n-5\}_{ht}$; red in~\cref{fig:e7}, 
\item $\{(a_8,b_j), j= n- 5, \dots, n\}_{ht}$; light red in~\cref{fig:e7},
\item $\{(a_i,b_8), i=\frac{n}{2},\dots, n\}_{ht}$; dark blue in~\cref{fig:e7},
\item $\{(a_{\frac{n}{2}-1},b_j), j =\frac{n}{2} , \dots,\frac{n}{2}+3)\}_{hh}$; gray in~\cref{fig:e7},
\item $\{(a_{n-8},b_j) , j = n-8, \dots, n \}_{hh}$; blue in~\cref{fig:e7},
\item $\{(a_i,b_{n-8}) , i = n-7, \dots, n \}_{hh}$; light blue in~\cref{fig:e7}.
\end{itemize}
 
\begin{figure*}[h!]
	\centering
	\includegraphics[page=5]{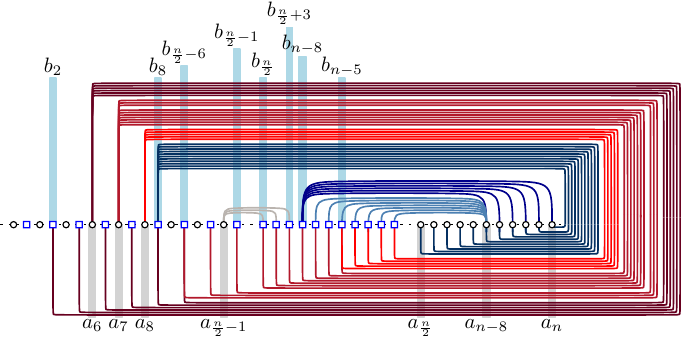}
    \caption{Page $7$ of $\mathcal{L}$ when $n$ is even.}
 \label{fig:e7}
\end{figure*}
 
\smallskip\noindent For $p = 8, \dots, \frac{n}{2}-7$, page $p$ of $\mathcal{L}$ contains the following \textcolor{blue}{$\frac{5n}{2}-2p+7$} edges:
 
\begin{itemize}[-]
\setlength\itemsep{0em}
\item $\{(a_{p-1},b_j), j= 2,\dots,\frac{n}{2}-p+1\}_{ht}$; dark red in~\cref{fig:ep},
\item $\{(a_p,b_j), j=\frac{n}{2} -p+1,\dots, n-p+2\}_{ht}$; red in~\cref{fig:ep}, 
\item $\{(a_{p+1},b_j), j=n-p+2, \dots, n\}_{ht}$; light red in~\cref{fig:ep},
\item $\{(a_i,b_{p+1}), i=\frac{n}{2},\dots, n\}_{ht}$; dark blue in~\cref{fig:ep},
\item $\{(a_i,b_{\frac{n}{2} +p-2}), i =\frac{n}{2} + (p-2) , \dots,n)\}_{hh}$; blue in~\cref{fig:ep},
\item $\{(a_{\frac{n}{2} + (p-2) },b_j), j = \frac{n}{2} + (p-1) ,\dots,n\}_{hh}$; light blue in~\cref{fig:ep}.
\end{itemize}
 
\begin{figure*}[h!]
	\centering
	\includegraphics[page=6]{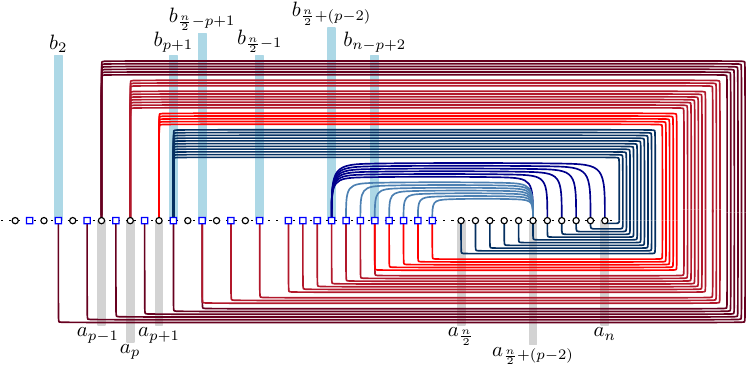}
    \caption{Page $p = 8, \dots, \frac{n}{2}-7$ of $\mathcal{L}$ when $n$ is even.}
 \label{fig:ep}
\end{figure*}
 
\smallskip\noindent Page $\frac{n}{2}-6$ of $\mathcal{L}$ contains the following \textcolor{blue}{2n-2} edges:
\begin{itemize}[-]
\setlength\itemsep{0em}
\item $\{(a_{\frac{n}{2}-7},b_j), j= 2,\dots,\frac{n}{2}-8\}_{ht}$; dark red in~\cref{fig:e},
\item $\{(a_{\frac{n}{2}-6},b_j), j= \frac{n}{2}-8,\dots, n-7\}_{ht}$; red in~\cref{fig:e}, 
\item $\{(a_{\frac{n}{2}-5},b_j), j= n-7, \dots, n\}_{ht}$; light red in~\cref{fig:e},
\item $\{(a_i,b_{\frac{n}{2}-5}), i=\frac{n}{2},\dots, n\}_{ht}$; blue in~\cref{fig:e},
\item $\{(a_i,b_{\frac{n}{2}- 4}), i = \frac{n}{2} + 5 , \dots,n)\}_{hh}$; dark blue in~\cref{fig:e}.
\end{itemize}
 
\begin{figure*}[h!]
	\centering
	\includegraphics[page=7]{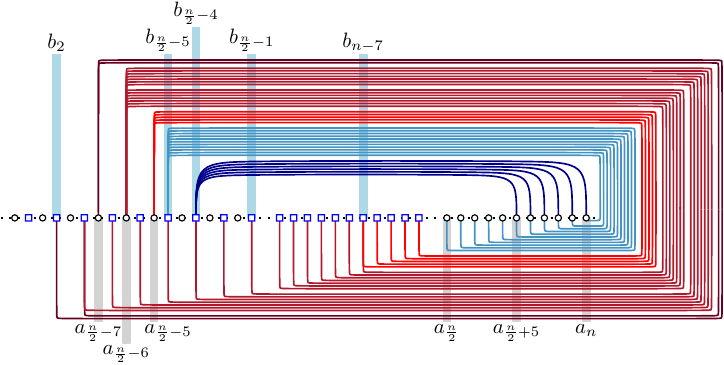}
    \caption{Page $\frac{n}{2}-6$ of $\mathcal{L}$ when $n$ is even.}
 \label{fig:e}
\end{figure*}
 
\smallskip\noindent For $p=\frac{n}{2} - k,$ with $ k \in \{5,4,3\}$, page $p$ of $\mathcal{L}$ contains the following \textcolor{blue}{$ \frac{3n}{2}-2k+22 $} edges:
 
\begin{itemize}[-]
\setlength\itemsep{0em}
\item $\{(a_{p-1},b_j), j= 2,\dots, \frac{n}{2}-p+1)\}_{ht}$; dark red in~\cref{fig:ek},
\item $\{(a_p,b_j), j=\frac{n}{2}-p+1,\dots, n-p+2)\}_{ht}$; red in~\cref{fig:ek}, 
\item $\{(a_{p+1},b_j), j= n-p+2, \dots, n\}_{ht}$; light red in~\cref{fig:ek},
\item $\{(a_i,b_{p+1}), i=\frac{n}{2},\dots,\frac{n}{2} + (k-1) \}_{ht}$; pink in~\cref{fig:ek},
\item $\{(a_i,b_{p+2}), i= \frac{n}{2} + (k-1),\dots, n-8+k \}_{ht}$; light pink in~\cref{fig:ek},
\item $\{(a_{n-8+k},b_j), j = \frac{n}{2}-3-k, \dots,\frac{n}{2}\}_{ht}$; blue in~\cref{fig:ek},
\item $\{(a_i,b_{n+(2k-11)}), i = n-8+k , \dots, n \}_{ht}$; light blue in~\cref{fig:ek},
\item $\{(a_i,b_{n+(2k-10)}), i = n-8+k , \dots,n\}_{hh}$; dark blue in~\cref{fig:ek},
\item $\{(a_{n-8+k},b_j), j = n+(2k-9) , \dots,n\}_{hh}$; gray in~\cref{fig:ek}.
\end{itemize}
 
\begin{figure*}[h!]
	\centering
	\includegraphics[page=8]{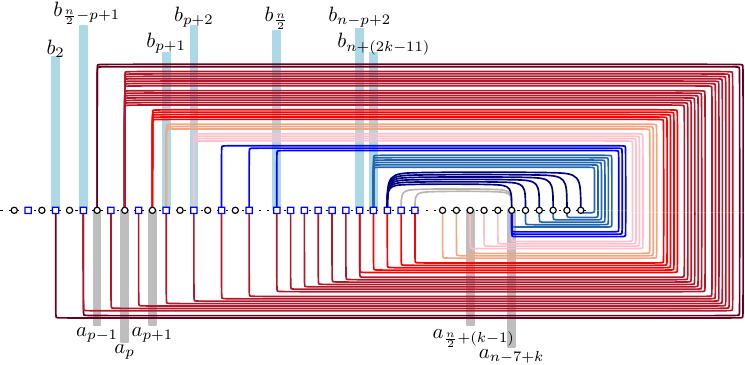}
    \caption{Page $p=\frac{n}{2} - k,$ with $ k \in \{5,4,3\}$ of $\mathcal{L}$ when $n$ is even.}
 \label{fig:ek}
\end{figure*}

\clearpage

\smallskip\noindent Page $\frac{n}{2} - 2$ of $\mathcal{L}$ contains the following \textcolor{blue}{$\frac{3n}{2}+19$} edges:
 
\begin{itemize}[-]
\setlength\itemsep{0em}
\item $\{(a_{\frac{n}{2}-2},b_j), j= 2,\dots, \frac{n}{2}+4\}_{ht}$; dark red in~\cref{fig:eL},
\item $\{(a_{\frac{n}{2}-1},b_j), j=\frac{n}{2}+4,\dots, n\}_{ht}$; red in~\cref{fig:eL}, 
\item $\{(a_i,b_{\frac{n}{2}-1}), i=\frac{n}{2},\frac{n}{2}+1\}_{ht}$; light red in~\cref{fig:eL},
\item $\{(a_i,b_{\frac{n}{2}}), i=\frac{n}{2}+1, \dots,n-7\}_{ht}$; pink in~\cref{fig:eL},
\item $\{(a_i,b_{n-7}), i=\frac{n}{4}+ 1, \dots, n\}_{ht}$; dark blue in~\cref{fig:eL},
\item $\{(a_i,b_{n-6}), i=\frac{n}{4}+ 1 , \dots, n\}_{hh}$; light blue in~\cref{fig:eL},
\item $\{(a_{\frac{n}{4}+ 1},b_j), j=n-5, \dots, n\}_{hh}$; gray in~\cref{fig:eL},
\item $\{(a_{\frac{n}{2}- 3},b_j), j=2,3\}_{hh}$; blue in~\cref{fig:eL}.
\end{itemize}
 
\begin{figure*}[h!]
	\centering
	\includegraphics[page=9]{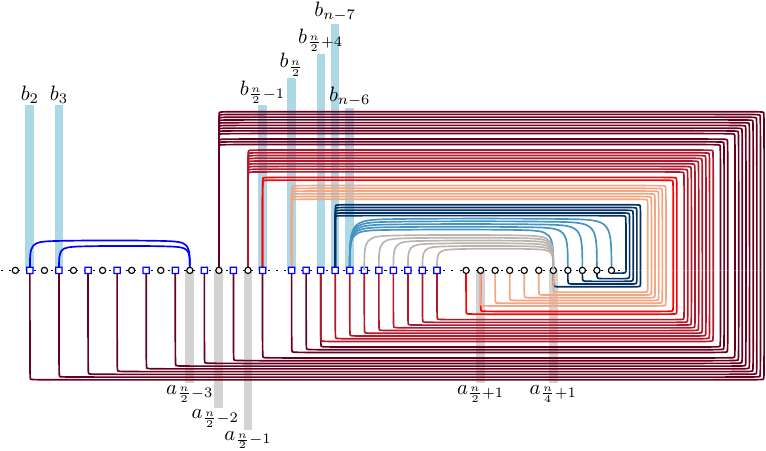}
    \caption{Page $p=\frac{n}{2}-2$ of $\mathcal{L}$ when $n$ is even.}
 \label{fig:eL}
\end{figure*}
 
\noindent So, in total $\mathcal{L}$ has $ n^2 $ edges. Since no two edges in the same rique deviate from the properties of cylindric layouts, it follows that the rique number of $K_{n,n}$ is at most $\frac{n}{2} -2 $ when n is even.

\begin{figure*}[p]
    \centering
    \includegraphics[page=1,width=\textwidth]{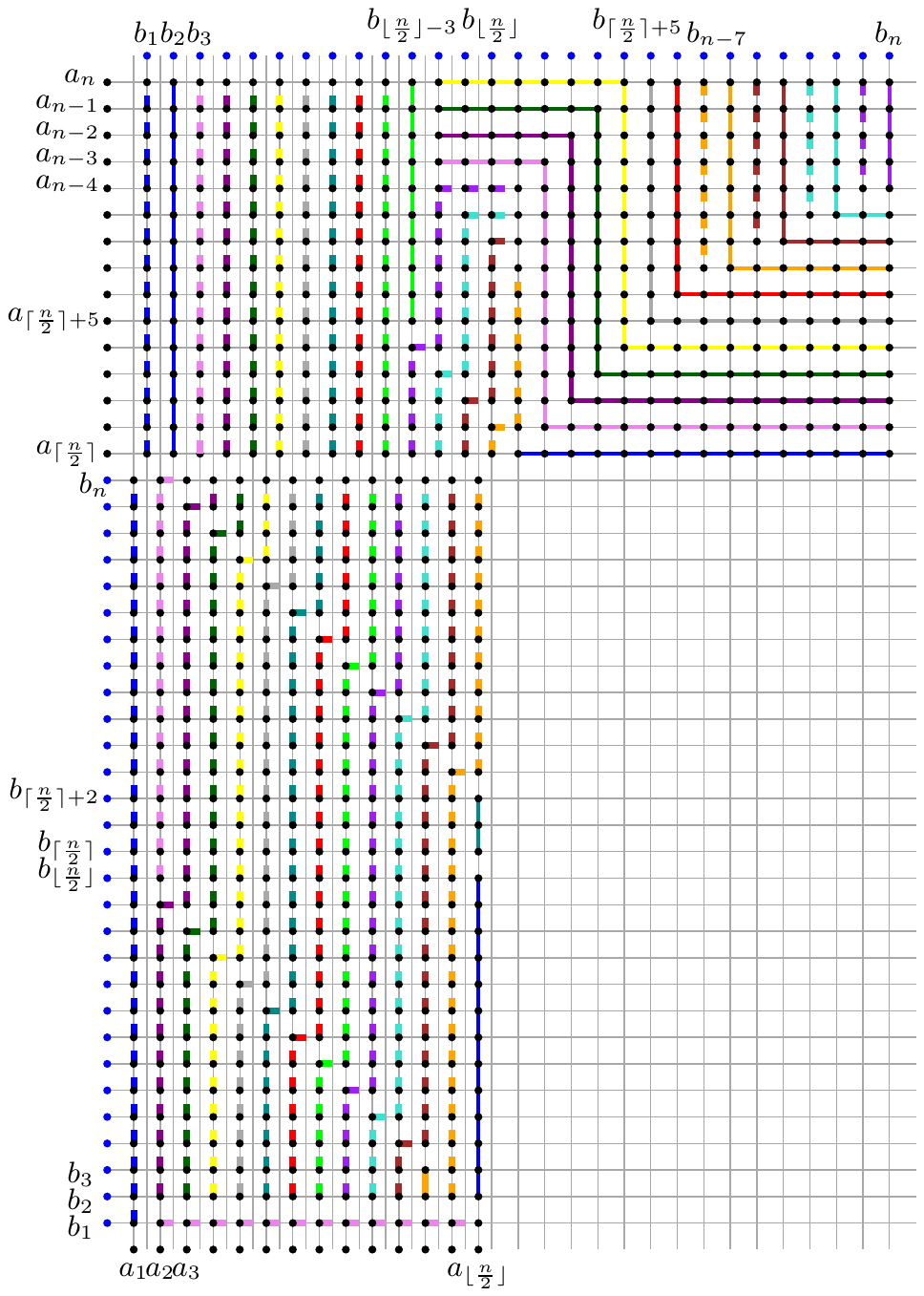}
    \caption{Illustration of the grid representation of a rique layout of $K_{n,n}$ when $n$ is odd, in which paths of the same color correspond to the same rique. The points of the grid that are covered by a solid (dashed) path are head-head (head-tail, respectively). 
    Here, the ``special'' pages are the first (blue), second (violet), third (dark-purple), fourth (dark-green), fifth (yellow), sixth (dark-gray), eighth (cyan), $\frac{n-1}{2}-5$ (green), $\frac{n-1}{2}-4$ (purple), $\frac{n-1}{2}-3$ (turquoise), $\frac{n-1}{2}-2$ (brown), and $\frac{n-1}{2}-1$ (orange) when $n$ is odd.}
    \label{fig:k29,29completebipartite}
\end{figure*}

\begin{figure*}
    \centering
    \includegraphics[page=2,width=\textwidth]{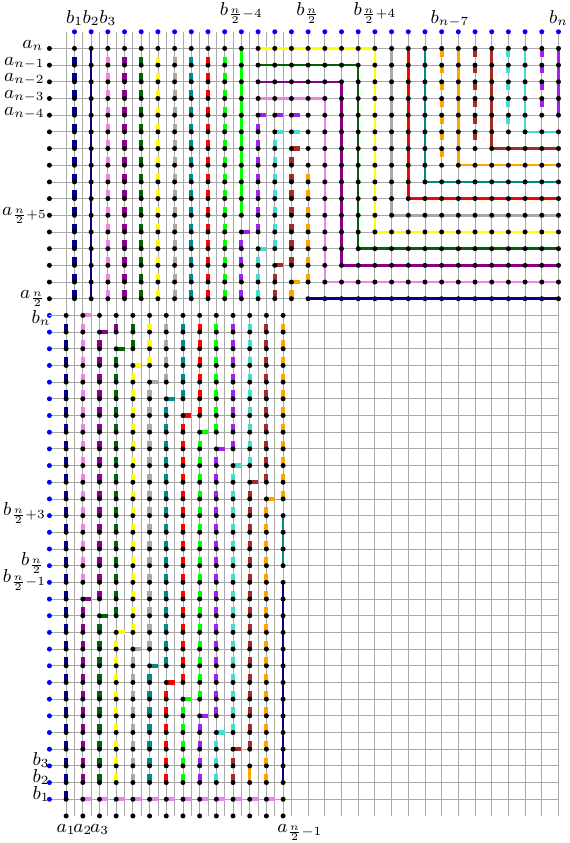}
    \caption{Illustration of the grid representation of a rique layout of $K_{n,n}$ when $n$ is even, in which paths of the same color correspond to the same rique. The points of the grid that are covered by a solid (dashed) path are head-head (head-tail, respectively). 
    When $n$ is even, the ``special'' pages are the first (blue), second (violet), third (dark-purple), fourth (dark-green), fifth (yellow), sixth (dark-gray), eighth (cyan), $\frac{n}{2}-6$ (green), $\frac{n}{2}-5$ (purple), $\frac{n}{2}-4$ (turquoise), $\frac{n}{2}-3$ (brown), and $\frac{n}{2}-2$ (orange).}
    \label{fig:k30,30completebipartite}
\end{figure*}

\end{proof}

\clearpage

\section{SAT formulation}
\label{sec:sat}

In this section, we present a SAT formulation for the problem of testing whether a given graph with $n$ vertices and $m$ edges admits a deque layout with $p$ of deques; an implementation has already been incorporated in~\cite{DBLP:journals/corr/abs-2003-09642}, whose source code is available at \url{https://github.com/linear-layouts/SAT}. However, before describing our formulation, we deem important to state that, in~\cite{DBLP:conf/gd/BekosFKKKR22}, Bekos et al.\ have already presented a corresponding SAT formulation, when the $p$ pages are riques. However, their approach heavily relies on the fact that this specific type of linear layouts can be characterized by means of a forbidden pattern in the underlying order (similar to the corresponding ones for stack and queue layouts~\cite{DBLP:conf/gd/Bekos0Z15}). Given that dequeue layouts cannot be characterized by means of such forbidden patterns in the underlying order, we need a slightly different approach.

Similar to~\cite{DBLP:conf/gd/BekosFKKKR22}, our approach is an extension of the one in~\cite{DBLP:conf/gd/Bekos0Z15} for the stack layout problem, in which there exist three different types of variables, denoted by $\sigma$, $\phi$, and $\chi$, with the following meanings: 
\begin{inparaenum}[(i)]
\item for a pair of vertices $u$ and $v$, variable $\sigma(u,v)$ is $\texttt{true}$, if and only if $u$ is to the left of $v$ along
the spine, 
\item for an edge $e$ and a page $i$, variable $\phi_i(e)$ is $\texttt{true}$, if and only if edge $e$ is assigned to page $i$ of the book, and 
\item for a pair of edges $e$ and $e'$, variable $\chi(e,e')$ is $\texttt{true}$, if and only if $e$ and $e'$ are assigned to the same page.
\end{inparaenum}  
Hence, there exist in total $O(n^2+m^2+pm)$ variables, while a set of $O(n^3)$ clauses ensures that the underlying order is indeed linear; for details see~\cite{DBLP:conf/gd/Bekos0Z15}. To overcome the issue that arises in the absence of forbidden pattern, we introduce $4pm$ variables, such that variable $\tau_i(e,x)$ with $x \in \{hh,ht,th,tt\}$ is $\texttt{true}$, if and only if the type of edge $e$ at page $i$ is $x$. We ensure that each edge has at least one of the allowed types, by introducing the following clause for each edge $e$:
 
\[
\bigvee_{i=1}^p (\tau_{i}(e,hh) \vee \tau_{i}(e,ht) \vee \tau_{i}(e,th) \vee \tau_{i}(e,tt))
\]
 
With these variables, we can express different configurations that cannot occur in a deque layout as clauses in the SAT formula. These clauses are obtained by avoiding crossings between all edge types in the cylindric representation of the graph. E.g., to express the different configurations that cannot occur for a head-head edge $e=(u,v)$ and a head-tail edge $e'=(u',v')$, we introduce the following clause for each page $i$ of the~layout\footnote{Note that some parts of the clause appear only certain conditions apply on the endpoints of $e$ and $e$'. These conditions are listed next to the corresponding parts, such that if a condition is not fulfilled, then the corresponding part has to be omitted.}:
 
\begin{align*} 
\phi_{i}(e) \wedge \phi_{i}(e') \wedge \tau_{i}(e,hh) \wedge \tau_{i}(e',ht)  & \rightarrow \\ 
~~~\neg (\sigma(u, u') \wedge \sigma(u', v) \wedge \sigma(v, v')) \quad u \neq v \neq u' &\\
~~~\wedge\neg (\sigma( v, u') \wedge \sigma( u', u) \wedge \sigma( u, v')) \quad u \neq v \neq u' &\\
~~~\wedge\neg (\sigma( u, v') \wedge \sigma( v', v) \wedge \sigma( v, u')) \quad u \neq v \neq v' &\\
~~~\wedge\neg (\sigma( v, v') \wedge \sigma( v', u) \wedge \sigma( u, u')) \quad u \neq v \neq v' &\\
~~~\wedge\neg (\sigma( u, u') \wedge \sigma( u', v') \wedge \sigma( v', v)) \quad u\neq u' &\\
~~~\wedge\neg (\sigma( u, v') \wedge \sigma( v', u') \wedge \sigma( u', v)) \quad u\neq v'& \\
~~~\wedge\neg (\sigma( v, u') \wedge \sigma( u', v') \wedge \sigma( v', u)) \quad v\neq u' & \\
~~~\wedge\neg (\sigma( v, v') \wedge \sigma( v', u') \wedge \sigma( u', u)) \quad v\neq v' &
\end{align*}

We introduce a clause similar to the one above for each pair of types of edges, yielding in total $O(pm^2)$ clauses. This completes the construction of the formula. Note that the formulation can be easily adjusted for rique layouts by introducing for each edge $e$ and each page $i$ the following clause forbidding tail-head and tail-tail edges: 
$\neg\tau_{i}(e,th) \wedge \neg\tau_{i}(e,tt)$.
 
Somehow unexpectedly, this simple adjustment was more efficient in practice than the one by Bekos et al.~\cite{DBLP:conf/gd/BekosFKKKR22}, which is based on implementing the forbidden pattern of rique layouts.

\smallskip\noindent\textbf{Findings.} The implementation was extremely helpful, in general, for developing all upper bounds of this paper. It further shows that the upper bound of \cref{thm:rique-complete} is tight for all values of $n \leq 30$ (see \cref{rem:bound-tightness}). Another notable observation is that for $K_{n,n}$ it is possible to obtain a better upper bound than the one of \cref{obs:deque-stack-bound} (or in other words that $k$ deques are strictly more powerful than $2k$ stacks): Our implementation shows that $K_{3n,3n}$ with $n \in \{2,3,4,5\}$ needs $n+1$ stacks, while the solver provided solutions with $n$ deques for the corresponding values of $n$. Note that this result would be implied (for any $n$) by \cref{thm:deque-bipartite}, if the bound~\cite{DBLP:journals/jct/EnomotoNO97} on the stack number of $K_{n,n}$ was shown to be tight. 

\section{Open Problems}
\label{sec:conclusions}

In this paper, we presented bounds on the deque- and rique-numbers of complete and complete bipartite graphs. 
We conclude with some open problems:  
 
\begin{inparaenum}[(i)]
\item We conjecture that the bound of \cref{thm:rique-complete} is tight.
\item As mentioned in the introduction, the deque-number of planar graphs is $2$. We conjecture that also their rique-number is $2$. 
\item Another natural direction to follow is to extend the study to other classes of graphs, as it is the case with the corresponding stack- and queue-numbers.
\item Studying inclusion relationships is also of interest, e.g., the class of graphs admitting 1-deques is not a subclass of the class of graphs admitting 1-rique, 1-stack layouts, as a maximal planar graph with a Hamiltonian path plus an edge belongs to the former but not to the latter. What about the other direction?
\item Related to our research is also the problem of closing the gap between the lower bound of $\lceil \frac{n}{2} \rceil$ and the upper bound of $\lfloor \frac{2n}{3} \rfloor+1$~\cite{DBLP:journals/jct/EnomotoNO97} on the stack number of $K_{n,n}$. 
\end{inparaenum}

\bibliographystyle{abbrv}
\bibliography{general, stacks, queues}
\end{document}